\documentclass[journal]{IEEEtran}
\usepackage{times} 
\usepackage{epsfig}
\usepackage{graphicx}
\usepackage{epstopdf}
\usepackage{algorithm,algorithmic,color,subfigure, latexsym, amsmath, amsfonts, amssymb, cite}
\usepackage{algorithm}
\usepackage{algorithmic}
\usepackage{epsfig}
\usepackage{graphicx}
\usepackage{epstopdf}
\usepackage{bm,amsmath,amssymb,amsfonts,graphicx,epsfig,amsthm,color}
\usepackage{bbold,dsfont}
\usepackage{float}
\usepackage{hhline}
\usepackage[hyphens]{url}
\PassOptionsToPackage{bookmarks=false}{hyperref}

\usepackage[depth=-1]{bookmark}

\usepackage{booktabs}       
\usepackage{amsfonts}       
\usepackage{microtype}      

\usepackage{times}
\usepackage{graphicx} 

\usepackage{wrapfig}



\usepackage{accents}
\makeatletter
\def\wid{\check{{\cc@style\underline{\mskip9.5mu}}}}
\def\Wideubar{\underaccent{{\cc@style\underline{\mskip6mu}}}}
\makeatother

\makeatletter
\def\wideubar{\underaccent{{\cc@style\underline{\mskip9.5mu}}}}
\def\Wideubar{\underaccent{{\cc@style\underline{\mskip6mu}}}}
\makeatother

\makeatletter
\def\widebar{\accentset{{\cc@style\underline{\mskip9.5mu}}}}
\def\Widebar{\accentset{{\cc@style\underline{\mskip6mu}}}}
\makeatother

\newtheorem{Lemma}{Lemma}

\newtheorem{Theorem}{Theorem}
\newtheorem{Definition}{Definition}
\newtheorem{Assumption}{Assumption}
\newtheorem{Example}{Example}

\theoremstyle{Remark}\newtheorem{Remark}{Remark}

\allowdisplaybreaks
\begin{document}
	
	\title{\Large \bf Model-Based and Data-Driven Control of Event- and Self-Triggered \\Discrete-Time LTI Systems}
	\author{Xin Wang,~Julian Berberich,~Jian Sun,
		Gang Wang, 
		Frank Allg{\"o}wer,
		and~Jie Chen
		\thanks{This work was supported in part by the National Key R$\&$D Program of
China under Grant 2021YFB1714800, the National Natural Science Foundation
of China under Grants 62088101, 61925303, 62173034,
and the Deutsche Forschungsgemeinschaft (DFG, German Research Foundation) under Germanys Excellence Strategy - EXC 2075 - 390740016 and under grant 468094890. We acknowledge the support by the Stuttgart Center for Simulation Science (SimTech).}
		\thanks{
			X. Wang, J. Sun, and G. Wang are with the Key Laboratory of Intelligent Control and Decision of Complex System, Beijing Institute of Technology, Beijing 10081, China.~J. Sun, and G. Wang are also with the Beijing Institute of Technology Chongqing Innovation Center, Chongqing 401120, China (e-mail: xinwang@bit.edu.cn; sunjian@bit.edu.cn; gangwang@bit.edu.cn).

J. Berberich and F. Allg{\"o}wer are with the University of Stuttgart, Institute for Systems Theory and Automatic Control, 70550 Stuttgart, Germany (e-mail: julian.berberich@ist.uni-stuttgart.de, frank.allgower@ist.uni-stuttgart.de).

J. Chen is with the Department of Control Science and Engineering, Tongji University, Shanghai 201804, China, and also with the State Key Lab of Intelligent Control and Decision of Complex Systems, School of Automation, Beijing Institute of Technology, Beijing 100081, China 	
(e-mail: chenjie@bit.edu.cn).}
	}
	\maketitle
	
	\begin{abstract}
The present paper considers the model-based and data-driven control of unknown linear time-invariant discrete-time systems under event-triggering and self-triggering transmission schemes.
To this end,
we begin by presenting a
dynamic event-triggering scheme (ETS) based on periodic sampling, and
a discrete-time looped-functional approach,
through which a model-based stability condition is derived.
Combining the model-based condition with a recent data-based system representation, a data-driven stability criterion in the form of linear matrix inequalities (LMIs) is established, which also offers a way of co-designing the ETS matrix and the controller.
To further alleviate the sampling burden of ETS due to its continuous/periodic detection, a self-triggering scheme (STS) is developed. Leveraging pre-collected input-state data, an algorithm for predicting the next transmission instant is given, while achieving system stability.
Finally, numerical simulations showcase the efficacy of ETS and STS in reducing data transmissions as well as of the proposed co-design methods.
\end{abstract}


\section{Introduction}
\label{sec:introduction}
Sampled-data control has received considerable attention in the study of networked control systems thanks to its convenience in system design and analysis. 
Traditional sampled-data schemes, which are mostly \emph{time-triggered}, feature low computational overhead and easy deployment,
but they may incur a sizable number of ``redundant" transmissions occupying network resources \cite{Donkers2012}.
Therefore, periodic sampling/transmission schemes are not appealing in networked control systems when communication resources (e.g., energy of wireless transmitting nodes, or network bandwidth) are limited.
Recently, research efforts have focused on developing transmission schemes that can use minimal communication resources while maintaining acceptable control performance; see, e.g., \cite{PENG2018113} for a survey.

A paradigm called the event-triggering scheme (ETS) \cite{Tabuada2007} has been proven efficient.
In ETS, the system state is continuously \cite{Tabuada2007}  or periodically \cite{Peng2013} monitored as in the time-triggered control, but measurements are transmitted only when deemed ``important".
This leads to a considerable decrease in resource occupancy while maintaining the system
performance.
In recent years, many variants of ETS have been proposed to further reduce transmissions; see, e.g.,
dynamic ETSs \cite{Girard2015,Liu2018}.
However, it is difficult for these ETSs to realize
continuous or periodic supervision of the system state.
To overcome this difficulty, the concept of self-triggering scheme (STS) was proposed.
Its core idea is to predict the next triggered instant based on a function constructed using the current sampled information as well as the system
knowledge.
Thus, in STS, the dedicated hardware in ETS is replaced with online \cite{Wang2010}, or offline \cite{Fiter2012} computations.
In networked control systems under STS, sensors are allowed to be completely shut off between adjacent sampling instants, leading to additional energy savings compared to the ETS.
Due to this advantage, the
STS  has been implemented in various fields \cite{Qi2020}. 
It should be mentioned that
most existing ETS and STS are designed for
continuous-time systems. 
However, in the context of networked control systems, continuous-time systems are often controlled via digital computers, in which case one typically first discretizes a continuous-time system and works with the resulting discrete-time counterpart.
Current literature has few results on discrete-time ETS/STS (e.g., \cite{Ding2020,Hu2016}).

All above-mentioned triggering
schemes are model-based, namely they require explicit knowledge of system models.
Nonetheless, obtaining accurate system models can be computationally demanding and oftentimes impossible in real-world applications.
Naturally,
an interesting question is \emph{how to co-design a controller and a triggering scheme without any knowledge of the system model}.
Measured data sequences, in practice, can be easily obtained.
One solution to the above question is to first estimate a model based on measured data, also known as system identification
\cite{Ljung1986}, and subsequently, perform model-based system analysis and controller design (see, e.g., \cite{Hou2013} for a survey).
Yet, such a two-stage approach comes with
an unavoidable drawback; that is, it is hard to provide an accurate model with guaranteed uncertainty bounds from limited and noisy data \cite{Matni2019,Oymak2019identification}.
An alternative approach, the so-called data-driven control, recently received increasing attention. Data-driven control
is aimed at learning control laws directly from data without resorting to any prior system identification steps.
Under this umbrella, various results  have been presented, including state-feedback and optimal control \cite{persis2020,van2020}, robust control \cite{berberich2020robust,Waarde2020}, control of time-delay systems \cite{rueda2020data},
predictive control \cite{Coulson2019}, and more can be found in the survey
\cite{markovskyabehavioral2021}.
By wedding the data-driven system representation in \cite{Waarde2020} with the model-based approach in \cite{Fridman2010}, a data-based stability condition for \emph{continuous-time} sampled-data control systems was derived in \cite{Berberich2020}, along with a controller design proposal. This data-driven framework has been extended to \emph{discrete-time} sampled-data systems \cite{wildhagen2021datadriven}. 
Data-driven ETS for continuous-time sampled-data systems with delays has been recently investigated in \cite{Wang2021data}.
It remains an untapped field to co-design a data-driven controller and ETS/STS for unknown discrete-time sampled-data systems.

These developments have motivated our work in this paper, which is focused on data-driven control of discrete-time sampled-data systems under ETS and STS.
As demonstrated in \cite{Girard2015}, under the dynamic ETS that contains an additional dynamic variable, the
triggering events can be reduced significantly compared to static ETS \cite{Tabuada2007} for continuous-time systems.
In this paper, we develop a discrete-time dynamic ETS based on periodic sampling, which is reminiscent of the ETS \cite{Peng2013} for continuous-time systems. 
For stability analysis of continuous-time sample-data systems, the looped-functional approach was proposed in \cite{Seuret2012} and subsequently explored by
\cite{Wang2021data}.
Looped-functionals often yield markedly improved stability conditions relative to common Lyapunov functionals,
since the looped-functional is only required to be monotonically decreasing at sampling points but not necessarily between these points.
 We generalize the looped-functional approach to discrete-time systems, by developing a discrete-time looped-functional (DLF) alternative, using which we derive model-based stability conditions for ETS.
 Combining this condition and the data-based representation of discrete-time systems in
 \cite{Waarde2020},
 a data-based stability condition is established, which provides a data-driven co-design method of the controller and ETS parameters.

On the other hand, a model-based discrete-time STS is designed, which can predict the next transmission instant 
without requiring online observation of state measurements between transmission times.
For \emph{unknown} discrete-time systems, it remains a key challenge to pre-compute the next execution time of sensor and controller using data, while ensuring stability under STS.
To address this issue, we rewrite the discrete-time sampled-data system as a switched system.
Using the data-driven parametrization of switched systems in \cite{wildhagen2021datadriven}, a data-driven algorithm for pre-computing the next transmission instant is derived, which does not require any explicit model knowledge.
Specially, the proposed STS law can be deduced to a special case of the dynamic ETS. Subsequently, the co-designed controller and triggering parameters under the ETS are employed to guarantee the stability of the system under the corresponding STS.
In a nutshell, the main contributions of the present paper are summarized as follows:
\begin{enumerate}
\item [\textbf{c1)}] A dynamic ETS based on periodic sampling for discrete-time systems, where the triggering condition depends on previously released data and current sampled data.
\item [\textbf{c2)}] Model- and data-based stability conditions for discrete-time systems under the dynamic ETS using a novel DLF approach, as well as model/data-driven  methods for co-designing the controller and triggering matrices; and,
\item [\textbf{c3)}] A model-based STS to predict
    the next transmission instant, and, building on this approach, a data-driven STS using only some pre-collected data from the system.
\end{enumerate}

While we only focus on closed-loop stability in this paper, it is straightforward to extend our results to obtain performance guarantees, e.g., on the closed-loop $\mathcal{L}_2$-gain, using similar arguments as in \cite{Peng2013}.

The rest of the paper is structured as follows.
In Section \ref{Sec:preliminaries}, we recall the problem setting as well as a discrete-time system representation which relies on noisy data.
In Section \ref{sec:eventtrigger}, an ETS control method for sampled-data systems is put forth, along with a DLF approach.
Then, an STS is developed in Section \ref{sec:self}. In both Sections \ref{sec:eventtrigger} and \ref{sec:self}, we present results for the model-based as well as the data-driven case.
Section \ref{sec:example} validates the merits and practicality of our methods and conditions using a numerical example. Section \ref{sec:conclude} draws concluding remarks.

{\it Notation.}
Throughout this paper, $\mathbb{N}$, $\mathbb{R^+}$, $\mathbb{R}^n$, and $\mathbb{R}^{n\times m}$ denote the sets of all non-negative
integers, non-negative real numbers, $n$-dimensional real vectors, and ${n\times m}$ real matrices, respectively. Then, we define 
$\mathbb{N}_{[a,b]}:=\mathbb{N}\cap [a,b]$, $a,b\in \mathbb{N}$.
We write $P\succ 0$ ($P\succeq 0$) if $P$ is a symmetric positive (semi)definite matrix;
${\rm diag}\{\cdots\}$ denotes a block-diagonal matrix;
${\rm Sym}\{P\}$ represents the sum of $P^{\top}$ and $P$.
We write $[\cdot]$ if elements in the matrix can be inferred from symmetry.
Let 
$0$ ($I$) denote zero (identity) matrices
of appropriate
dimensions.
Notation `$\ast$' represents the symmetric term in (block) symmetric matrices.
We use $\|\cdot\|$ to stand for the Euclidean norm of a vector.

\section{Preliminaries}\label{Sec:preliminaries}

Consider the following discrete-time linear system 
\begin{equation}\label{sec1:sys:disLTI}
x(t+1)=A x(t)+B u(t),~x(0)=x_0\in \mathbb{R}^{n}
\end{equation}
for $t\in \mathbb{N}$, where $x(t)\in \mathbb{R}^{n}$ is the system state, $u(t)\in \mathbb{R}^{m}$ is the control input, and $A\in\mathbb{R}^{n\times n}$, $B\in\mathbb{R}^{n\times m}$ are the system matrices.
Such discrete-time systems can model networked control systems (NCSs) equipped with digital devices, where the sensor, the controller, and the actuator act at discrete instants.
We consider in this paper that the system matrices $A$ and $B$ are \emph{unknown}, but some pre-collected state-input measurements $\{x(T)\}^{\rho}_{T=0}$ and $\{u(T)\}^{\rho-1}_{T=0}$ $(T\in \mathbb{N},~\rho\in \mathbb{N}_{[1,\infty]})$ satisfying the following dynamic
\begin{equation}\label{sys:data:perturbed}
x(T+1)= Ax(T)+B u(T)+B_ww(T)
\end{equation}
are available at discrete time instants $T \in \{0,1,\ldots,\rho\}$. Here, $B_w\in \mathbb{R}^{n\times n_w}$ is a known matrix, which has full column rank and models the influence of the disturbance on the collected data.
The measured data are corrupted by an {\it unknown} noise (perturbation) sequence $\{w(T)\}^{\rho-1}_{t=0}$, where $w(T)\in \mathbb{R}^{n_w}$ captures, e.g., process noise or unmodeled system dynamics. This noise only affects the data generated for the controller design and will be neglected in the closed-loop operation of our ETS and STS scheme, but we note that an extension in this direction is straightforward.
The available measurements can be stacked to form the following data matrices

\begin{align*}
X_+ &:=\big[\begin{array}{cccc}x(1)&x(2)&\cdots &x(\rho) \\\end{array}\big],\\
X&:=\left[\begin{array}{cccc}x(0)&x(1)&\cdots &x(\rho-1) \\\end{array}\right],\\
U&:=\left[\begin{array}{cccc}u(0)&u(1)&\cdots &u(\rho-1) \\\end{array}\right],\\
W&:=\left[\begin{array}{cccc}w(0)&w(1)&\cdots &w(\rho-1) \\\end{array}\right]
\end{align*}
where $X_+$, $X$, and $U$ are known, but $W$ is unknown.
Then,
it is evident that
\begin{align}\label{formu:data}
X_+=AX+BU+B_wW.
\end{align}

For the proposed approach, the required data are collected offline and can be transmitted once by the sensor for the controller design step. Therefore, we can neglect network effects on these data and we can assume that they are collected with sampling period $1$.
In practice, the noise is typically bounded.
We make the following standing assumption on the noise. 
\begin{Assumption}[\emph {Noise bound}]\label{Ass:disturbance}
\emph{The noise sequence $\{w(T)\}^{\rho-1}_{T=0}$ collected in the matrix $W$ belongs to
\begin{align}\label{data:disturbance}
\mathcal{W}=\bigg\{W\in\mathbb{R}^{n_w\times\rho} \Big |
\left[\!\begin{array}{cc}W^{\top}\\I \\\end{array}\!\right]^{\top}
  \left[\!\begin{array}{cc}Q_d\! & \!S_d\\\ast\! & \!R_d \\\end{array}\!\right]
  \left[\!\begin{array}{cc}W^{\top}\\I \\\end{array}\!\right]\succeq0 \bigg\}
\end{align}
for some known matrices $Q_d \prec 0 \in \mathbb{R}^{\rho\times \rho}$, $S_d \in \mathbb{R}^{\rho\times n_w}$, and $R_d=R_d^{\top} \in \mathbb{R}^{n_w\times n_w}$.}
\end{Assumption}
Assumption \ref{Ass:disturbance} provides a general framework  to model bounded additive noise, which has been used in similar forms in \cite{persis2020,Berberich2020,Waarde2020,wildhagen2021datadriven}.

Based on Equation \eqref{formu:data} and Assumption \ref{Ass:disturbance}, we define $\Sigma_{AB}$ to be the set of all pairs $[A~B]$ adhering to the measured data and the noise bound, namely
\begin{align*}
\Sigma_{AB}:=\Big\{[A~ B] \Big | X_+=AX+BU+B_wW,~ W\in \mathcal{W}\Big\}.
\end{align*}
Then, an equivalent expression of $\Sigma_{AB}$ is provided in the form of a quadratic matrix inequality (QMI). 
\begin{Lemma}[{\emph {Data-based representation} \cite[Lemma 4]{Waarde2020}}] \label{Lemma:system:data} The set $\Sigma_{AB}$ is equal to
\begin{align*}
\Sigma_{AB}=\bigg\{[A~B]\in\mathbb{R}^{n\times (n+m)} \Big |
\left[\begin{array}{cc}[A~B]^{\top}\\I \\\end{array}\right]^{\top}
  \Theta_{AB}[\cdot]^{\top}\succeq0
\bigg\}
\end{align*}
where $\Theta_{AB}:=
\left[\begin{array}{cc}-X & 0 \\ -U & 0 \\ \hline X_+ & B_w \\\end{array}\right]
  \left[\begin{array}{cc}Q_d & S_d\\\ast & R_d \\\end{array}\right]
  [\cdot]^{\top}.$
\end{Lemma}
Lemma \ref{Lemma:system:data} provides a purely data-based representation for the system \eqref{sec1:sys:disLTI} with unknown matrices $A$ and $B$ using only data $X_+$, $X$ and $U$.
In order to guarantee the stability of \eqref{sec1:sys:disLTI} without relying on any
knowledge of the system matrices, we need to achieve a
stability criterion for all $[A~B]\in \Sigma_{AB}$.

Note that, in Fig. \ref{FIG:structure:event}, data are collected offline in
an open-loop experiment.
In online closed-loop operation, the system state is sampled and transmitted to the controller at time $t_k \in \mathbb{N}$, where
$t_0=0$, $t_{k+1}-t_k\geq1$, $k \in \mathbb{N}$.
In the controller, the sampled state $x(t_k)$ is available, and the control input is computed via the linear state-feedback law $u(t_k)=Kx(t_k)$ (and it is held constant until $t_{k+1}-1$), where $K$ is the controller gain matrix to be designed.
The system \eqref{sec1:sys:disLTI} under the closed-loop sampled-data control can be written as
\begin{equation}\label{sys:sampling}
x(t+1)=A x(t)+BKx(t_k),~~t\in \mathbb{N}_{[t_k, t_{k+1}-1]}.
\end{equation}
Traditional periodic transmission schemes have been used for data-driven control, e.g., by \cite{Berberich2020,wildhagen2021datadriven,wildhagen2021improved}, which determine the maximum sampling interval for which stability can be guaranteed.
Avoiding ``redundant'' transmissions in networks,
we develop a model-based ETS and STS for system \eqref{sys:sampling} to adaptively determine the transmission instant $t_k$ to save communication resources. 
Subsequently, we provide methods for data-driven co-design of ETS/STS and the corresponding controller based on the model-based stability conditions, where Lemma \ref{Lemma:system:data} is employed to describe the system matrices consistent with the data. The recent paper \cite{Wang2021data} addresses data-driven event-triggered control for continuous-time systems. In contrast, in the present paper, we address \emph{discrete-time} systems and event-triggered \emph{as well as self-triggered control}.

\section{Event-Triggered Control}\label{sec:eventtrigger}
In this section, we propose an ETS for the discrete-time system in Section \ref{subsec:event}.
as well as a novel looped-functional for discrete-time sampled-data system in Section \ref{SEC:DLF}. The stability of the event-triggered system is analyzed in Sections \ref{sec:model:ETS}. Based on this, Section \ref{design:event} derives a data-driven analysis and co-design method of the controller and the ETS without any explicit knowledge of the  system matrices.


\subsection{Discrete-time dynamic ETS}\label{subsec:event}
\begin{figure}
	\centering
 \includegraphics[scale=0.5]{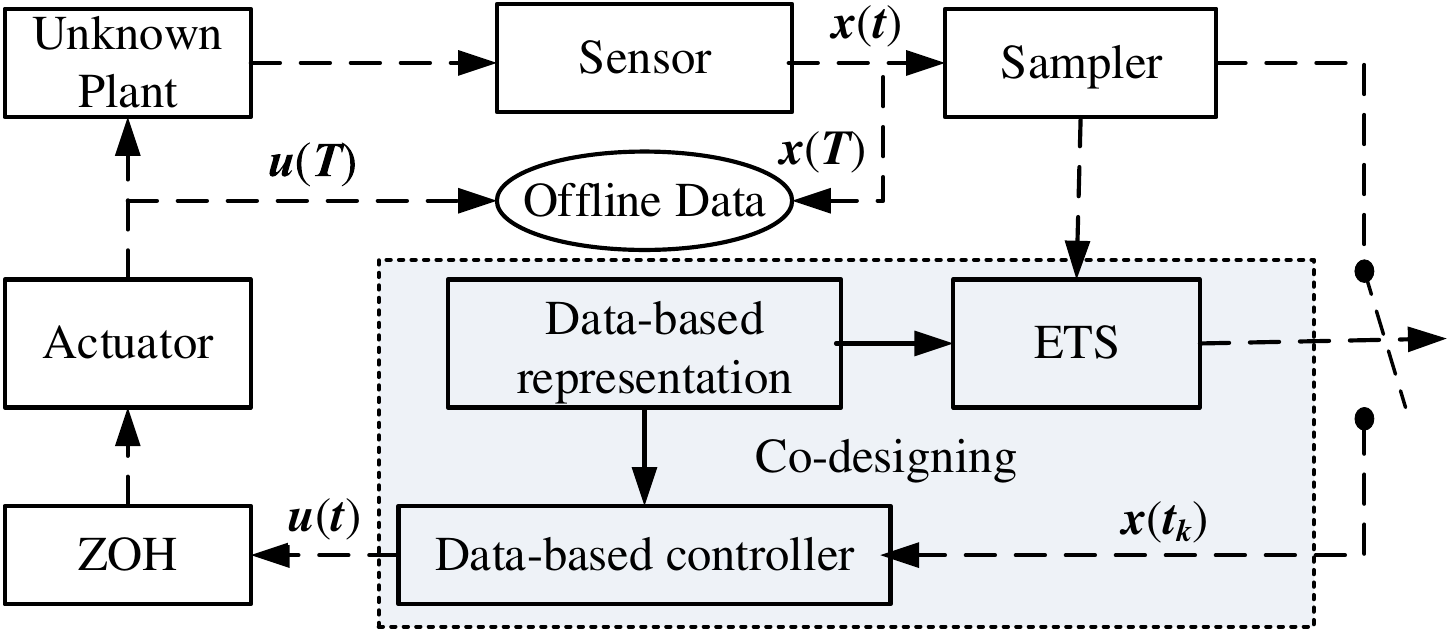}
	\caption{Structure of data-driven sampled-data systems under ETS.}
	\label{FIG:structure:event}
\end{figure}

\begin{figure}
	\centering
\includegraphics[scale=0.4]{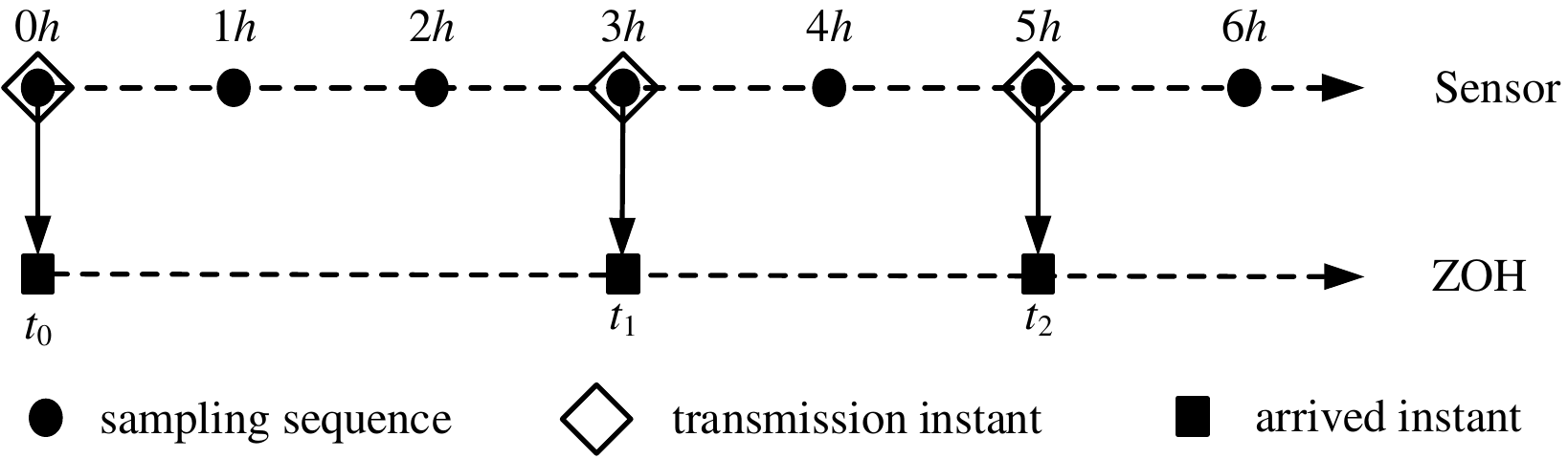}
	\caption{Evolution of sampling and transmission events.}
	\label{FIG:evolution:event}
\end{figure}
We use a discrete-time dynamic event-triggering  module to dictate the transmission instants $\{t_k\}_{k\in \mathbb{N}}$, as depicted in Fig. \ref{FIG:structure:event}.
In our dynamic ETS, the system states are firstly
periodically sampled at discrete instants $\{t_k+j h\}_{j,k \in \mathbb{N}}$, where the sampling interval $h$ is a constant satisfying $1\leq \underline{h} \leq h \leq \bar{h}$ for given lower and upper bounds $\underline{h},~\bar{h} \in \mathbb{N}$. 
In the next step of the ETS,
the sampled signal $x(t_k+j h)$ is checked for the following triggering law
\begin{align}\label{sys:judgement}
&\eta(\tau_j^k)+\theta \rho(\tau_j^k)<0,
\end{align}
where $\theta>0$ is to be designed; $\tau_j^k:=t_k+jh$ for all $j \in \mathbb{N}_{[0,m_k]}$ with  $m_k=\frac{t_{k+1}-t_k}{h}-1$; $\rho({\tau_j^k})$ is a
discrete-time function defined by
\begin{equation}\label{trigger:function}
\rho({\tau_j^k}):=\sigma_1 x^{\top}({\tau_j^k})\Omega x({\tau_j^k})+\sigma_2 x^{\top}(t_k)\Omega x(t_k)-e^{\top}({\tau_j^k})\Omega e({\tau_j^k})
\end{equation}
where $\Omega\succ0$ is some weight matrix; $\sigma_1\geq0$ and $\sigma_2\geq0$ are triggering parameters to be designed; $e({\tau_j^k}):=x({\tau_j^k})-x(t_k)$ denotes the error between sampled signals $x({\tau_j^k})$ at the current sampling instant and $x(t_k)$ at the latest transmission instant;
and,
$\eta(t)$ is a dynamic variable, satisfying
$\eta(t)=\eta({\tau_j^k})$ for $t \in \mathbb{N}_{[{\tau_j^k}, {\tau_{j+1}^k}-1]}$ and
the following difference  equation
\begin{align}\label{sys:dynamic}
\eta({\tau_{j+1}^k})-\eta({\tau_j^k})=-\lambda\eta({\tau_j^k})+\rho({\tau_j^k}),
\end{align}
where $\eta(0)\geq0$ and $\lambda>0$ are given parameters.

Once condition \eqref{sys:judgement} is violated, the triggering module sends the sampled state $x({\tau_j^k})$ to the controller, and a new control input is computed which is held via a zero-order hold (ZOH) element in
the interval $[t_{k+1}, t_{k+2}-1]$. Subsequently, the triggering module is updated using the latest transmitted data, and monitors the next sampled system state.
In summary, our discrete-time dynamic ETS can be formulated as
\begin{align}\label{sys:trigger}
&t_{k+1}=t_k+h\cdot\min_{j\in \mathbb{N}}\Big\{j>0\Big|\eta({\tau_j^k})+\theta\rho({\tau_j^k})<0\Big\}.
\end{align}
In Fig. \ref{FIG:evolution:event}, an example is presented to illustrate the proposed triggering transmission scheme. On the side of the sensor, discrete-time data are sampled periodically at a constant sampling period $h$. Then, the sampled data are sent to the controller at the instants $0h$, $3h$, and $5h$, when the condition \eqref{sys:judgement} is satisfied, but none of the others.

The following lemma states that the extra dynamic variable satisfies $\eta({\tau_j^k})\geq 0$, $\forall t\in \mathbb{N}$. under the condition \eqref{sys:trigger} with the given parameters $\eta(0)\geq0$, $\theta>0$, and $\lambda>0$. The proof of Lemma \ref{lemma:nonneg.dynam} is similar to
\cite[Lemma 2]{Wang2021data}, which is omitted here.
\begin{Lemma}[\emph {Non-negativity}]\label{lemma:nonneg.dynam}
Let $\eta(0)\ge0$, 
$\Omega\succ 0$, 
and $\lambda>0$, $\theta> 0$ be constants satisfying $1-\lambda-\frac{1}{\theta}\geq0$. 
Then, 
it holds that $\eta({\tau_j^k})\geq0$, $\forall j\in \mathbb{N}_{[0,m_k]}, k\in \mathbb{N}$, under the condition \eqref{sys:trigger}.
\end{Lemma}

\begin{table}[t]
\caption{Hyperparameters and description.}
\begin{center}      
\setlength{\tabcolsep}{4pt}
\renewcommand\arraystretch{1.}
\begin{tabular}{llccccccccc}
\hline\noalign{\smallskip}
Hyperparameter & Description \\
\noalign{\smallskip}\hline\noalign{\smallskip}
$h \in \mathbb{N}_{[1, \infty)}$ & Sampling interval\\
$\sigma_1,\sigma_2 \in [0,1]$ & Triggering threshold parameters\\
$\lambda,\theta\in (0,\infty)$ & Triggering parameters\\
$\eta({\tau_j^k})\in [0,\infty)$& Triggering dynamic variable\\
$\Omega\succ0$ & Triggering matrix\\
\hline\noalign{\smallskip}
\end{tabular}
\end{center}
\label{Tab:hyperparameters}
\end{table}

\begin{Remark}\label{general}
\emph{The transmission scheme \eqref{sys:trigger} is a discrete-time analog of the continuous-time dynamic ETS proposed by \cite{Wang2021data}.
Seen from \eqref{sys:trigger}, our ETS becomes the dynamic ETS proposed in \cite{Hu2016} by setting $h=1$ and $\sigma_2=0$; when, in addition, the parameter $\theta$ approaches infinity, the condition \eqref{sys:trigger} further boils down to
the static ETS \cite{Liu2015,Tripathy2017Discrete}; moreover, it degenerates to the time-triggering scheme \cite{Seuret2012} when $\theta$ tends to infinity, $\sigma_1=0$, and $\sigma_2=0$.
Thus, our triggering scheme
in \eqref{sys:trigger} unifies and generalizes several existing ETS, and is expected to further reduce
data transmissions 
and save
transmission resources, which is shown in Section \ref{sec:example}.}
\end{Remark}

\subsection{Discrete-time looped-functional approach}\label{SEC:DLF}
\begin{figure}
	\centering
		 \includegraphics[scale=0.49]{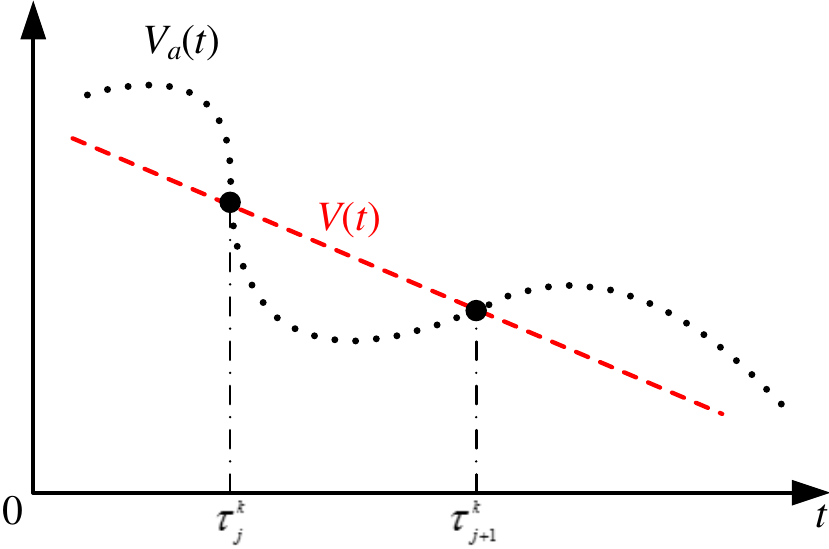}
	\caption{Illustration of Lemma \ref{lemma:DLF}. The function $V_a(t)$ tends to zero since $\vartriangle V(t)<0$ and $V_l(x({\tau_{j+1}^k}),{\tau_{j+1}^k}) = V_l(x({\tau_j^k}),{\tau_j^k})$.}
	\label{FIG:DLF}
\end{figure}
This section develops a discrete-time looped-functional (DLF) approach for stability analysis.
To begin with, the definition of DLF is given below.
\begin{Definition}
[\emph {DLF}]\label{def:DLF}
\emph{Consider any series of time instants $\{{\tau_j^k}\}$. A functional
$V_l(x,t):\mathbb{R}^n\times \mathbb{N}_{[{\tau_j^k}, {\tau_{j+1}^k}-1]} \rightarrow \mathbb{R}$ is called a {\it DLF}, if it satisfies
\begin{equation}\label{def:inequality}
	V_l(x({\tau_{j+1}^k}),{\tau_{j+1}^k}) = V_l(x({\tau_j^k}),{\tau_j^k}),~\forall j\in \mathbb{N}.
\end{equation}}
	\end{Definition}
Since the property that $V_l(x({\tau_j^k}),{\tau_j^k})= V_l(x({\tau_{j+1}^k}),{\tau_{j+1}^k})$ shares the characteristics of the looped-functional in \cite{Seuret2012} but is defined at discrete time instants, we call $V_l(x(t),t)$ a DLF.
Similar to the continuous looped-functional approach \cite[Theorem 1]{Seuret2012},
the following stability theorem can be obtained by using the (discontinuous) {DLF}.
\begin{Lemma} \label{lemma:DLF}
Choose a function $V_a:\mathbb{R}^n\rightarrow \mathbb{R^+}$ with scalars $c_2>c_1>0$, and $p>0$
satisfying $\forall x\in \mathbb{R}^{n}, ~c_1 ||x||^p \leq V_a(x) \leq c_2 ||x||^p$.
Then the following statements are equivalent.
\begin{enumerate}
\item The increment of the function $V_a$ is strictly negative at all $\{{\tau_j^k}\}$,
i.e., $\Delta V_a:=V_a(x({\tau_{j+1}^k}))-V_a(x({\tau_j^k}))<0, ~\forall x({\tau_j^k})\neq 0$.
\item There exists a {{DLF}} $V_l$, such that for all $t\in [{\tau_j^k}, {\tau_{j+1}^k}-1]$, 
    $\Delta V(t):=V(t+1)-V(t)<0,~\forall x({\tau_j^k})\neq 0$,
    and $\{{\tau_j^k}\}$,
where $V(t):=V_a(x(t))+V_l(x(t),t)$.
\end{enumerate}
Moreover, if either one of the two statements is true, the origin of the system \eqref{sys:sampling} is asymptotically stable.
\label{lemma:GELF}
\end{Lemma}
\begin{proof}
Let $j\in \mathbb{N}_{[0,m_k]}$, $k\in \mathbb{N}$, and $t\in \mathbb{N}_{[{\tau_j^k}, {\tau_{j+1}^k}-1]}$.
	{\emph {2) $\Rightarrow$ 1)}}.
	{Assume that 2) is satisfied}. Summing up
	$\Delta V(t)$ over $[{\tau_j^k}, {\tau_{j+1}^k}-1]$ yields
	\begin{align*}
		\sum \limits_{t={\tau_j^k}}^{{\tau_{j+1}^k}-1} &\Delta V(t)
		=\sum \limits_{t={\tau_j^k}}^{{\tau_{j+1}^k}-1} \Big[V_a(x(t+1))-V_a(x(t))\Big]\\
          & +\sum \limits_{t={\tau_j^k}}^{{\tau_{j+1}^k}-1} \Big[V_l(x(t+1),t+1)-V_l(x(t),t)\Big]<0.
	\end{align*}
	According to \eqref{def:inequality}, it follows that $\Delta V_a<0$.\\
	{\emph {1) $\Rightarrow$ 2)}}.
	Assume 1) is satisfied. Similar to \cite[Theorem 1]{Seuret2012}, the following functional is introduced for $t\in \mathbb{N}_{[{\tau_j^k}, {\tau_{j+1}^k}-1]}$
	\begin{align*}	
     V_l(x(t),t)=-V_a(x(t))+\frac{t}{{\tau_{j+1}^k}-{\tau_j^k}}\Delta V_a
	\end{align*}
    which satisfies \eqref{def:inequality}.
	Then, we have that
	\begin{align*}
		\Delta V(t)&=V_a(x(t+1))-V_a(x(t))-V_a(x(t+1))\\
&~~~~+V_a(x(t))+\frac{t+1}{{\tau_{j+1}^k}-{\tau_j^k}}\Delta V_a-\frac{t}{{\tau_{j+1}^k}-{\tau_j^k}}\Delta V_a\\
&=\frac{1}{{\tau_{j+1}^k}-{\tau_j^k}}\Delta V_a<0.
	\end{align*}
	This proves the equivalence between 1) and 2).\\
	{\emph {Asymptotic stability}}. From the condition 1), we have that
$||x({\tau_j^k})||\rightarrow 0$ as $k\rightarrow \infty$.
Finally, similar to \cite{Fuj2009,Briat2012}, there exist $\delta<\infty$ yielding $||x(t)|| \leq \delta ||x({\tau_j^k})||$ for all $t\in \mathbb{N}_{[{\tau_j^k}, {\tau_{j+1}^k}-1]}$,  which implies that the system \eqref{sys:sampling} is asymptotically stable under 1) or 2).
\end{proof}

\begin{Remark}\emph{Lemma \ref{lemma:DLF} was inspired by the continuous-time looped-functional approach \cite{Seuret2012}.
If compared to the discrete-time Lyapunov stability theorem, Lemma \ref{lemma:DLF} provides less conservative stability conditions by constructing a proper DLF for discrete-time sampled-data systems, since the added DLF is not necessarily a positive definite functional.
This can be clearly seen from Fig. \ref{FIG:DLF}.
Note that the discrete Lyapunov functional $V_a(t)$ is not required to decrease at each time $t$, 
but only at ordered and non-adjacent discrete points $t_k$.
To the best of our knowledge, a DLF as in Lemma \ref{lemma:DLF} has not yet been introduced in discrete-time sampled-data control systems.}
\end{Remark}

\subsection{Model-based stability analysis}\label{sec:model:ETS}
In this subsection, we develop a model-based stability condition for the sampled-data system \eqref{sys:sampling} under the transmission scheme \eqref{sys:trigger}, where matrices $A$ and $B$ are assumed \emph{known}.
 Our triggering strategy \eqref{sys:trigger} comprises a time-trigger with sampling interval $h$.
Increasing values of $h$ lead to fewer data transmissions under our triggering strategy, where $h$ is directly related to the stability properties. Based on this observation,
we derive a stability criterion for system \eqref{sys:sampling} under \eqref{sys:trigger} using the DLF approach. Before moving on, a useful result is given.
\begin{Lemma}[\emph {Summation inequality}]\label{lemma:summation}
For any vector $\vartheta\in\mathbb{R}^m$, matrices $R=R^{\top}\in\mathbb{R}^{n\times n}\succ0$, $N\in\mathbb{R}^{m\times 2n}$, scalars $\alpha\leq\beta \in \mathbb{N}$, and 
a sequence $\{x(s)\}_{s=\alpha}^{\beta-1}$,
the following summation inequality holds true
\begin{equation*}\label{Lemma:inequality}
\begin{aligned}
-\sum \limits_{i=\alpha}^{\beta-1} y^{\top}(i)R y(i)
\leq (\beta\!-\!\alpha)\vartheta^{\top} \!N \mathcal{R}^{-1}\! N^{\top}\!\vartheta
\!+\!{\rm Sym}\left\{\vartheta^{\top} N \Pi\right\}
\end{aligned}
\end{equation*}
where $y(i):=x(i+1)-x(i),~\mathcal{R}:={\rm diag}\left\{R, \,3R\right\}$, $\Pi:=\!\Big[x^{\top}(\beta)\!-\!x^{\top}(\alpha), x^{\top}(\beta)\!+\!x^{\top}(\alpha)\!- \!\sum \limits_{i=\alpha}^{\beta} \frac{x^{\top}(i)}{\beta-\alpha+1}\Big]^{\top}$.
\end{Lemma}
Lemma \ref{lemma:summation} can be cast as a special case of \cite[Lemma 2]{Chen2017summation}, whose proof is omitted here. Based on Lemmas \ref{lemma:DLF} and \ref{lemma:summation}, we have the following model-based stability condition.

\begin{Theorem}[\emph {Model-based condition}]\label{Th1}
For given scalars $\bar{h}>\underline{h}>0$, $ \sigma_{1}\geq0$, $\sigma_{2}\geq0$, $\lambda>0$, and $\theta> 0$ satisfying $1-\lambda-\frac{1}{\theta}\geq0$, asymptotic
stability of system \eqref{sys:sampling} is achieved under the triggering condition \eqref{sys:trigger}, and $\eta({\tau_j^k})$ converges to the origin for $\eta(0)\geq 0$, if there exist matrices $P\succ0$, $R_1\succ0$, $R_2\succ0$, $\Omega\succ0$,
$S$, $N_1$, $N_2$, $F$, such that the following LMIs hold for $h\in\{\underline{h}, \bar{h}\}$
\begin{align}{\label {Th1:LMI1}}
&\left[
  \begin{array}{ccc}
    \Xi_0+h\Xi_\varsigma+\Psi+\mathcal{O}  & hN_\varsigma\\
    \ast & -h\mathcal{R}_\varsigma
  \end{array}
\right]\prec0,~\varsigma=1,2
\end{align}
where $\Psi:={\rm Sym}\big\{F(AL_1+BKL_{7}-L_{2})\big\}$, and
\begin{align*}
\Xi_0&:={\rm Sym}\Big\{\Pi_1^{\top} S \Pi_2 \!-\! \Pi_3^{\top} S \Pi_4 \!+\! N_1\Pi_9 \!+\! N_2 \Pi_{10} \Big\}\!+\!L_2^{\top} P L_2\\
&~~~~- L_1^{\top} P L_1+(L_2-L_1)^{\top} (R_2-R_1) (L_2-L_1)\\
\Xi_1&:={\rm Sym}\Big\{\Pi_5^{\top} S \Pi_6 \Big\}+(L_2-L_1)^{\top} R_2(L_2-L_1)\\
\Xi_2&:={\rm Sym}\Big\{\Pi_7^{\top} S \Pi_8 \Big\}+(L_2-L_1)^{\top} R_1(L_2-L_1)\\
\mathcal{O}&:=\sigma_1 L_3^{\top} \Omega L_3+
\sigma_2 L_{7}^{\top}\Omega L_{7} - (L_3-L_{7})^{\top} \Omega(L_3-L_{7})\\
\mathcal{R}_1&:={\rm diag}\left\{R_1, \,3R_1\right\},~
\mathcal{R}_2:= {\rm diag}\left\{R_2, \,3R_2\right\}\cr
\Pi_1&:=\left[L_3^{\top},\, L_{4}^{\top},\,L_{2}^{\top}-L_3^{\top},\,L_{5}^{\top}+L_2^{\top}-L_3^{\top}\right]^{\top}\\
\Pi_2&:=\left[-L_3^{\top},\, -L_4^{\top},\, L_4^{\top}-L_2^{\top}, \,L_6^{\top}-L_1^{\top}-L_4^{\top}\right]^{\top}\\
\Pi_3&:=\left[L_0^{\top},\, L_0^{\top},\, L_1^{\top}-L_3^{\top}, \,L_5^{\top}-L_3^{\top}\right]^{\top}\\
\Pi_4&:=\left[L_0^{\top}, \,L_0^{\top},\, L_4^{\top}-L_1^{\top},\, L_6^{\top}-L_4^{\top}\right]^{\top}\\
\Pi_5&:=\left[L_3^{\top}, \,L_4^{\top},\, L_0^{\top},\, L_5^{\top}\right]\\
\Pi_6&:=\left[-L_3^{\top}, \,-L_4^{\top},\, L_1^{\top}-L_2^{\top},\, -L_1^{\top}\right]^{\top}\\
\Pi_7&:=\left[L_3^{\top}, \,L_4^{\top},\, L_2^{\top}-L_1^{\top},\, L_2^{\top}\right]^{\top}\\
\Pi_8&:=\left[L_3^{\top}, \,L_4^{\top},\, L_0^{\top},\, L_6^{\top}\right]^{\top}\\
\Pi_9&:=\left[L_1^{\top}-L_3^{\top},\, L_1^{\top}+L_3^{\top}-2L_5^{\top}\right]^{\top}\\
\Pi_{10}&:=\left[L_4^{\top}-L_1^{\top}, \,L_4^{\top}+L_1^{\top}-2L_6^{\top}\right]^{\top}\\
L_i&:=\left[0_{n\times (i-1)n}, \,I_n, \,0_{n\times (7-i)n} \right], \;(i=1, 2,\ldots,7)\\
L_0&:=0_{n\times7n}.
\end{align*}
\end{Theorem}
\begin{proof}
Considering the augmented system state $(x,\eta)$ for $t \in \mathbb{N}_{[{\tau_j^k}, {\tau_{j+1}^k}-1]}$ 
with $j\in \mathbb{N}_{[0,m_k]}$ and $k\in \mathbb{N}$,
we choose 
\begin{align}{\label {Th1:Vt}}
V(x,t)=V_a(x(t))+V_l(x(t),t)+
\eta(t)
\end{align}
where $V_a(x(t))=x^{\top}(t)Px(t)$, $P \succ 0$;   $\eta(t)$ is given in \eqref{sys:dynamic} and $\eta({\tau_j^k})\geq0$ due to Lemma \ref{lemma:nonneg.dynam};
moreover, $V_l(x(t),t)$ is a novel DLF given as
\begin{equation}\label{Th1:Wl}
V_l(x(t),t)=\sum \limits_{i=1}^{3}V_{li}(t)
\end{equation}
with $V_{l1}(t):=2\phi_1^{\top}(t)S \phi_2(t)$, and
\begin{align*}
V_{l2}(t)&:=({\tau_{j+1}^k}-t)\Big[\sum \limits_{i={\tau_j^k}}^{t} y^{\top}(i)R_1y(i) -y^{\top}(t)R_1y(t) \Big], \\
V_{l3}(t)&:=(t-{\tau_j^k})\Big[\sum \limits_{i=t}^{{\tau_{j+1}^k}} y^{\top}(i)R_2y(i) -y^{\top}(t)R_2y(t) \Big],
\end{align*}
where $S$, $R_1\succ 0$, $R_2\succ 0$; and, $y(i):=x(i+1)-x(i)$, $\phi_0:=\big[x^{\top}({{\tau_j^k}})$, $x^{\top}({{\tau_{j+1}^k}})\big]^{\top}$, $\phi_1(t):=\big[(t-{\tau_j^k})\phi_0^{\top}$, $x^{\top}(t)-x^{\top}({{\tau_j^k}}), \sum \limits_{i={\tau_j^k}}^{t}x^{\top}(i)-x^{\top}({\tau_j^k})\big]^{\top}$, $\phi_2(t):=\big[({\tau_{j+1}^k}-t)\phi_0^{\top}$, $x^{\top}({\tau_{j+1}^k})-x^{\top}(t),\sum \limits_{i=t}^{{\tau_{j+1}^k}}x^{\top}(i)-x^{\top}({\tau_{j+1}^k})\big]^{\top}$.

According to Lemma \ref{lemma:DLF}, we calculate the forward difference  $\Delta V(x,t):= V(x(t+1),t+1)-V(x(t),t)$, yielding
\begin{align}{\label {Th1:Vd}}
\Delta V(x,t)= \Delta V_a(t)+ \Delta V_l(t)+ \Delta\eta(t)
\end{align}
where $\Delta\eta(t)=\eta(t+1)-\eta(t)$, 
and
\begin{align*}
\Delta V_a(t)&=\xi^{\top}(t)\big(L_2^{\top} P L_2 - L_1^{\top} P L_1 \big)\xi(t),\\
\Delta V_{l1}(t)&=2\xi^{\top}(t)\big[\Pi_1^{\top} S \Pi_2 - \Pi_3^{\top} S \Pi_4+(t-{\tau_j^k})\Pi_5^{\top} S \Pi_6 \\
&~~~~~~~~~~~~~+ ({\tau_{j+1}^k}-t)\Pi_7^{\top} S \Pi_8 \big]\xi(t),\\
\Delta V_{l2}(t)&=\xi^{\top}(t)\big[({\tau_{j+1}^k}\!-\!t\!-\!1)(L_2\!-\!L_1)^{\top} R_1(L_2\!-\!L_1)\big] \xi(t)\\
&~~~-\sum \limits_{i={\tau_j^k}}^{t-1} y^{\top}(i)R_1y(i),\\
\Delta V_{l3}(t)&=\xi^{\top}(t)\big[(t-{\tau_j^k}+1)(L_2-L_1)^{\top} R_2(L_2-L_1)\big] \xi(t)\\
&~~~-\sum \limits_{i=t}^{{\tau_{j+1}^k}-1} y^{\top}(i)R_1y(i),
\end{align*}
where the notation $\xi(t)$ is given as $\xi(t):=\big[x^{\top}(t)$, $x^{\top}(t+1)$, $x^{\top}({\tau_j^k}), x^{\top}({\tau_{j+1}^k}), \sum \limits_{i={\tau_j^k}}^{t}\frac{x^{\top}(i)}{t-{\tau_j^k}+1},
\sum \limits_{i=t}^{{\tau_{j+1}^k}}\frac{x^{\top}(i)}{{\tau_{j+1}^k}-t+1}, x^{\top}(t_k)\big]^{\top}.$

By Lemma \ref{lemma:summation}, the summation terms \eqref{Th1:Vd} satisfy
\begin{equation}
\begin{aligned}
-&\sum \limits_{i={\tau_j^k}}^{t-1} y^{\top}(i)R_1y(i)
-\sum \limits_{i=t}^{{\tau_{j+1}^k}-1} y^{\top}(i)R_2y(i)\leq \\
&~\xi^{\top}(t)\big[(t-{\tau_j^k})N_1 \mathcal{R}_1^{-1}N_1^{\top}+({\tau_{j+1}^k}-t)N_2 \mathcal{R}_2^{-1}N_2^{\top}\\
&~~~~~~~+2N_1\Pi_9+2N_2\Pi_{10}\big]\xi(t).
\end{aligned}
\end{equation}

Through the descriptor method \cite{Fridman2010}, the model-based system representation \eqref{sys:sampling} can be written as, for $t \in \mathbb{N}_{[{\tau_j^k}, {\tau_{j+1}^k}-1]}$
\begin{align}{\label {Th1:zero}}
0
=2\xi^{\top}(t)F \big(AL_1+BKL_{7}-L_{2}\big)\xi(t),
\end{align}
where $F$ is a fixed matrix of dimensions $7n\times n$.

In light of \eqref{sys:trigger}, when currently sampled data are not
transmitted, Lemma \ref{lemma:nonneg.dynam} asserts that $\eta({\tau_j^k})\geq0$ for $\lambda>0$, $\eta(0)\ge0$, and $\theta>0$ satisfying $1-\lambda-\frac{1}{\theta}\geq0$. Hence, from \eqref{sys:dynamic} and \eqref{trigger:function}, it holds that
\begin{equation}\label{th1:trigger}
\eta({\tau_j^k})-\eta({\tau_{j+1}^k})=-\lambda\eta({\tau_j^k})+\rho({\tau_j^k}) \leq  \xi^{\top}(t) \mathcal{O} \xi(t). 
\end{equation}

Consequently, we have that by summing up \eqref{Th1:Vd}-\eqref{th1:trigger}
 \begin{align}
\Delta & V(x,t)+\eta({\tau_j^k})-\eta({\tau_{j+1}^k})-\Delta\eta(t)\notag\\
&\leq \xi^{\top}(t)\left[\frac{t-{\tau_j^k}}{h}\Upsilon_1(h) +\frac{{\tau_{j+1}^k}-t}{h}\Upsilon_2(h) \right]\xi(t) \label {Th1:sum}
\end{align}
where $\Upsilon_\varsigma(h)=\Xi_0+\Psi+\mathcal{O}+h\Xi_{\varsigma} +h N_\varsigma \mathcal{R}_\varsigma^{-1} N_\varsigma^{\top}$ for $\varsigma=1,2.$

Using the Schur Complement Lemma, it can be deduced that $\Upsilon_1(h)\prec0$ and $\Upsilon_2(h)\prec0$ are equivalent
to the LMIs in \eqref{Th1:LMI1},  which are affine, and consequently convex, with respect to $h$. Thus, LMIs \eqref{Th1:LMI1} at the vertices of $h\in[\underline{h},\bar{h}]$ ensure $\Delta V(x,t)+\eta({\tau_j^k})-\eta({\tau_{j+1}^k})-\Delta\eta(t)<0$ for all $h\in[\underline{h},\bar{h}]$.
It follows from Lemma \ref{lemma:DLF} that
\begin{equation}{\label {Th1:vj}}
\sum \limits_{t={\tau_j^k}}^{{\tau_{j+1}^k}-1} \Big[\Delta V(x,t)+ \eta({\tau_{j+1}^k})-\eta({\tau_j^k})-\Delta\eta(t) \Big]<0,\forall x({\tau_j^k})\neq 0
\end{equation}
which ensures $V_a({\tau_{j+1}^k})+(h-1)\eta({\tau_{j+1}^k})<V_a({\tau_j^k})+(h-1)\eta({\tau_j^k})$, $\forall j\in \mathbb{N}_{[0,m_k]}$ and $k\in \mathbb{N}$.
We conclude that System \eqref{sys:sampling} and the dynamic variable $\eta({\tau_j^k})$ converge to the origin asymptotically under our
transmission scheme when $k \rightarrow \infty$, since $V_a(t)>0$, $\eta({\tau_j^k})>0$, and $x(t)$ is bounded during $t\in \mathbb{N}_{[{\tau_j^k}, {\tau_{j+1}^k}-1]}$, thereby completing the proof.
\end{proof}

\begin{Remark}[\emph {Discussion}]
\emph{Theorem \ref{Th1} provides a stability criterion for the discrete-time system \eqref{sys:sampling} under the dynamic triggering scheme \eqref{sys:trigger} using the DLF approach in Lemma \ref{lemma:DLF}.
The latter is a discrete-time version of the looped-functional approach \cite{Seuret2012} that has been applied to deduce less conservative stability conditions.
In \cite{Hu2016,Liu2015}, related stability criteria have been presented for the event-triggered control of discrete-time systems. However, under these results, the lower bound of the inter-execution interval is $t_{k+1}-t_k=1$.
Recently, a switching dynamic event-triggered control method for discrete-time systems was proposed by \cite{Ding2020}, where a guaranteed lower bound larger than $1$, i.e., $t_{k+1}-t_k>1$, is beneficial for reducing the amount of transmissions.
Our periodic-sampling-based dynamic ETS \eqref{sys:trigger} explained in Section \ref{subsec:event} only requires that $t_{k+1}-t_k\geq h$, where $h$ is the periodic sampling interval. 
Based on Theorem \ref{Th1}, 
 we can determine a possibly large value of $h$ leading to stability, thus saving communication resources.}

\end{Remark}
\begin{Remark}[\emph {DLF}]
\emph{A special case of DLF in the form of quadratic matrix functions is employed in \eqref{Th1:Wl}. It can be proven that
$V_l(x({\tau_j^k}),{\tau_j^k})=V_l(x({\tau_{j+1}^k}),{\tau_{j+1}^k})=0$, which satisfies
condition \eqref{def:inequality} in Definition \ref{def:DLF}.
In \eqref{Th1:Wl}, only the system state and its single-summation terms 
are considered. Less conservative stability criteria can be derived if higher-order summation terms, e.g., multiple-summation of the system state, are included. This is left for future research.}
\end{Remark}

\subsection{Data-based stability analysis and controller design}\label{design:event}
We now derive a data-based stability certificate for the event-triggered control system \eqref{sys:sampling} with \emph{unknown} system matrices $A$ and $B$, as well as a data-driven method for co-designing the controller gain $K$ and the triggering matrix $\Omega$.
Motivated by \cite{Berberich2020},
the main idea is to employ a system expression using the data $\{x(T)\}^{\rho}_{T=0}$, $\{u(T)\}^{\rho-1}_{T=0}$ to replace the model-based representation in \eqref{sec1:sys:disLTI}.
Following this line,
the data-based system representation in Lemma \ref{Lemma:system:data}, combined with the model-based
stability condition in Theorem \ref{Th1}, is employed to obtain a data-based stability condition.
We begin with an algebraically equivalent system to \eqref{sys:sampling}.

Assume that $G\in \mathbb{R}^{n \times n}$ is nonsingular, and let $x(t)=Gz(t)$. The system \eqref{sys:sampling} is restructured as follows
\begin{equation}\label{Design:NCS}
z(t+1)=G^{-1}AG z(t)+G^{-1}BK_c z(t_k)
\end{equation}
for $t\in \mathbb{N}_{[t_k, t_{k+1}-1]}$, where $K_c:=KG$.
The system \eqref{Design:NCS} exhibits the same stability behavior as \eqref{sys:sampling}, and the triggering condition \eqref{sys:trigger}
remains effective.
Based on Theorem \ref{Th1}, we have the following theoretical result.
\begin{Theorem}
[\emph {Data-driven condition and controller
design}]\label{Th2}
For given the same scalars $\bar{h}>\underline{h}>0$, $ \sigma_{1}\geq0$, $\sigma_{2}\geq0$, $\epsilon$, $\lambda>0$, and $\theta>0$ satisfying $1-\lambda-\frac{1}{\theta}\geq0$, there exists a controller gain $K$ such that asymptotic stability of  system \eqref{sys:sampling} is achieved under the triggering condition \eqref{sys:trigger} for any $[A ~B]\in \Sigma_{AB}$, and $\eta({\tau_j^k})$ converges to the origin, if there exist a scalar $\varepsilon>0$, and matrices $P\succ0$, $R_1\succ0$, $R_2\succ0$, $\Omega_z\succ0$,
$S$, $N_1$, $N_2$, $G$, $K_c$, such that the following LMIs hold for all $h\in\{\underline{h}, \bar{h}\}$
\begin{align}{\label {Th2:LMI1}}
&\left[
  \begin{array}{ccc}
     \mathcal{T}_1& \mathcal{T}_2+\mathcal{F}& 0\\
    \ast &  \mathcal{T}_3+\Xi_0+h\Xi_\varsigma+\bar{\Psi}+\bar{\mathcal{O}}  & hN_\varsigma\\
     \ast &  \ast & -h\mathcal{R}_\varsigma
  \end{array}
\right]\prec0
\end{align}
where $\varsigma=1,2$, $\bar{\Psi}:={\rm Sym}\big\{-\mathcal{D}GL_{2}\big\}$, and
\begin{align*}
\bar{\mathcal{O}}&:=\sigma_1 L_3^{\top} \Omega_z L_3+
\sigma_2 L_{7}^{\top}\Omega_z L_{7}- (L_3-L_{7})^{\top} \Omega_z(L_3-L_{7})\\
\mathcal{D}&:=(L_1+\epsilon L_2)^{\top},~\mathcal{F}:=\big[L_1^{\top}G^{\top},\,L_{7}^{\top}K_c^{\top} \big]^{\top}\\
\mathcal{T}_1&:=\varepsilon\mathcal{V}_1\Theta_{AB}\mathcal{V}_1^{\top},~
\mathcal{T}_2:=\varepsilon\mathcal{V}_1\Theta_{AB}\mathcal{V}_2^{\top},~
\mathcal{T}_3:=\varepsilon\mathcal{V}_2\Theta_{AB}\mathcal{V}_2^{\top}\\
\mathcal{V}_1&:=
\left[\begin{array}{ccc}I & 0\\\end{array}\right],~
\mathcal{V}_2:=
\left[\begin{array}{ccc} 0& \mathcal{D}\\\end{array}\right].
\end{align*}
Moreover, $K=K_cG^{-1}$ and $\Omega={G^{-1}}^{\top} \Omega_z G^{-1}$ are desired the controller gain $K$ and triggering matrix.
\end{Theorem}
\begin{proof}
Replacing $x$ of the functional in \eqref{Th1:Vt} with the state $z$, the following functional is built for the system \eqref{Design:NCS}
\begin{align}{\label {Th2:Vt}}
V(z,t)=\!V_a(z(t))+V_l(z(t),t)+\eta(t).
\end{align}

Using the descriptor method again, the system \eqref{Design:NCS} can be represented as
\begin{align}{\label {Th:controller:AB:zero}}
0
=2\xi_z^{\top}(t)\mathcal{D}(AGL_1+BK_cL_7-GL_2)\xi_z(t)
\end{align}
where $\epsilon$ is a given constant and $\xi_z(t):=\big[z^{\top}(t)$, $z^{\top}(t+1), z^{\top}({\tau_j^k}), z^{\top}({\tau_{j+1}^k}),
\sum \limits_{i={\tau_j^k}}^{t}\frac{z^{\top}(i)}{t-{\tau_j^k}+1},
\sum \limits_{i=t}^{{\tau_{j+1}^k}}\frac{z^{\top}(i)}{{\tau_{j+1}^k}-t+1}, z^{\top}(t_k)\big]^{\top}$.

Triggering condition \eqref{sys:trigger} and Lemma \ref{lemma:nonneg.dynam} prove that $\eta({\tau_{j+1}^k})-\eta({\tau_j^k}) \leq  \xi^{\top}(t) \mathcal{O} \xi(t)$ as in \eqref{th1:trigger}, which directly ensures the following inequality with $x(t)=Gz(t)$
 \begin{align}\label{Th4:Vd}
\eta({\tau_{j+1}^k})-\eta({\tau_j^k}) \leq  \xi_z^{\top}(t) \bar{\mathcal{O}} \xi_z(t).
\end{align}

Imitating \eqref{Th1:sum}, 
it can be obtained that
 \begin{align*}
\Delta &V(z,t)+\eta({\tau_j^k})-\eta({\tau_{j+1}^k})-\Delta\eta(t)\\
&\leq \xi_z^{\top}(t)\left[\frac{t-{\tau_j^k}}{h}\bar{\Upsilon}_1(h) +\frac{{\tau_{j+1}^k}-t}{h}\bar{\Upsilon}_2(h) \right]\xi_z(t)
\end{align*}
where $\bar{\Upsilon}_\varsigma(h):={\rm Sym}\big\{\mathcal{D}(AGL_1+BK_cL_{7})\big\}+\Xi_0+h\Xi_\varsigma+\bar{\Psi}+\bar{\mathcal{O}}+
h N_\varsigma \mathcal{R}_\varsigma^{-1} N_\varsigma^{\top}$ for  $\varsigma=1,2$.
Then, the terms $\bar{\Upsilon}_1(h)$ and $\bar{\Upsilon}_2(h)$ are restructured as follows
\begin{align*}
\bar{\Upsilon}_\varsigma(h)&:=
\left[\begin{array}{cc}[\mathcal{D}A~\mathcal{D}B]^{\top}\\I \\\end{array}\right]^{\top}\\
&~~~~~\times\left[\begin{array}{cc}0 & \mathcal{F}\\\ast & \Xi_0+h\Xi_\varsigma+\bar{\Psi}+\bar{\mathcal{O}}+
hN_\varsigma \mathcal{R}_\varsigma^{-1} N_\varsigma^{\top} \\\end{array}\right]\!
\left[\cdot\right].
\end{align*}

According to the data-based representation in Lemma \ref{Lemma:system:data}, it holds for any $[A ~ B]\in\Sigma_{AB}$ that
\begin{equation}
\begin{aligned}
\left[\!\begin{array}{cc}[A~B]^{\top}\\I \\\end{array}\!\right]^{\top}
  \!\Theta_{AB}\!
  \left[\!\begin{array}{cc}[A~B]^{\top} \\
   I\\\end{array}\!\right]\succeq0.
\end{aligned}
\end{equation}
By the full-block S-procedure \cite{Sche2001}, we have $\bar{\Upsilon}_1(h)\prec0$ and $\bar{\Upsilon}_2(h)\prec0$ for any $[A ~B]\in\Sigma_{AB}$ if there exists a scalar $\varepsilon>0$ such that for $\varsigma=1,2$
 \begin{align}\label{Th2:fullblock1}
&\left[\begin{array}{cc}0 & \mathcal{F}\\\ast & \Xi_0+h\Xi_\varsigma+\bar{\Psi}+\bar{\mathcal{O}}
+hN_\varsigma \mathcal{R}_\varsigma^{-1} N_\varsigma^{\top}\\\end{array}\right]\notag\\
&~~~~~~~~~~~~~~~~~~+
\varepsilon \left[\begin{array}{cc}\mathcal{V}_1\Theta_{AB}\mathcal{V}_1^{\top} & \mathcal{V}_1\Theta_{AB}\mathcal{V}_2^{\top}\\\ast &  \mathcal{V}_2\Theta_{AB}\mathcal{V}_2^{\top} \\\end{array}\right] \prec0.
\end{align}

Finally, the Schur Complement Lemma ensures that the inequalities in \eqref{Th2:fullblock1} are equivalent to LMIs in \eqref{Th2:LMI1}. Similar to Theorem \ref{Th1}, we have a conclusion that LMIs \eqref{Th2:LMI1} at the vertices of $h\in[\underline{h},\bar{h}]$ are sufficient stability conditions for system \eqref{Design:NCS}
under the triggering condition \eqref{sys:trigger} for any
$[A ~B]\in \Sigma_{AB}$, and that $\eta({\tau_j^k})$ converges to the origin.
Since $G$ is nonsingular, system \eqref{Design:NCS} exhibits the same stability behavior as \eqref{sys:sampling}.
\end{proof}

\begin{Remark}[\emph {Model-based design under ETS}]\label{remark:model:design}
\emph{A model-based co-design method under the ETS \eqref{sys:trigger} can be derived
by replacing \eqref{Th2:LMI1} with the condition $\bar{\Upsilon}_\varsigma(h)\prec0$ for $\varsigma=1,2$.}
\end{Remark}

\begin{Remark}[\emph {Novelty}]
\emph{Up to date, there are few research efforts devoted to data-driven control of discrete-time systems under aperiodic sampling, 
in particular
data-driven control under event-triggering schemes. Theorem \ref{Th2}, which is based on Theorem \ref{Th1},
provides a stability condition and a
co-design method of the controller and the ETS for unknown sampled-data control systems. A possibly large sampling interval $h$ and a
triggering matrix $\Omega$ for the triggering condition \eqref{sys:trigger} can be
searched for 
by using Theorem
\ref{Th2}.
The application of Theorem \ref{Th2} is simple, requiring only a solution of LMIs which can be constructed based on noisy data.
Besides,
when $\theta$ approaches infinity, $\sigma_1=0$, and $\sigma_2=0$, the triggering condition \eqref{sys:trigger} degenerates to the periodic sampling scheme studied in \cite{wildhagen2021datadriven,wildhagen2021improved}.
In comparison to the results in \cite{wildhagen2021datadriven,wildhagen2021improved}, Theorem \ref{Th2} provides larger values of $h$,
which shows the role of the DLF in reducing conservatism of stability conditions. The comparison results are listed in Table \ref{Tab:noise:givenK} for a numerical example.}
\end{Remark}

\section{Self-Triggered Control}\label{sec:self}
As highlighted in Section \ref{sec:introduction}, STS does not rely on extra hardware to continuously monitor the system states \cite{Wan2021self},
but rather predicts the next sampling instant based on a local function and previous data. 
In this section, we study data-driven STS with \emph {unknown} matrices $A$ and $B$, as depicted in Fig. \ref{FIG:structure:self}. The challenge here is to predict the next transmission instant using the already transmitted system states and historical noisy measurements (i.e., $\{x(T)\}^{\rho}_{T=0}$, $\{u(T)\}^{\rho-1}_{T=0}$ as in Section \ref{Sec:preliminaries}) without explicit knowledge of the system matrices $A$ and $B$. To this end, we begin by designing a model-based STS.
\subsection{Model-based STS}
In order to apply data-driven control arguments similar to those in Section \ref{sec:eventtrigger}, we first need to define a lifted version of the original system \eqref{sys:sampling} as suggested in \cite{wildhagen2021datadriven}. To this end, let us define for $s>0$, $s\in \mathbb{N}$
\begin{equation*}
\begin{aligned}
\underline{B}^s&:=
\left[\begin{array}{cccc}A^{s-1}B&A^{s-2}B&\cdots &B \\\end{array}\right],\\
\underline{K}^s&:=\Big[~\underbrace{K^{\top} ~~ K^{\top} ~~\cdots ~~K^{\top}}_{s~ \text{times}}\Big]^{\top}.
\end{aligned}
\end{equation*}
We exploit that discrete-time sampled-data systems can be viewed as switched systems, which is a well-known fact in the literature \cite{Yu2004} 
\begin{equation}\label{sys:switch}
x(t_k+s_k)=(A^{s_k}+\underline{B}^{s_k} \underline{K}^{s_k}) x(t_k),~~s_k\in \mathbb{N}_{[1, \,\bar{s}]}
\end{equation}
where $s_k=t_{k+1}-t_{k}$ and $\bar{s}>1\in \mathbb{N}$.

\begin{figure}[t]
	\centering
		 \includegraphics[scale=0.5]{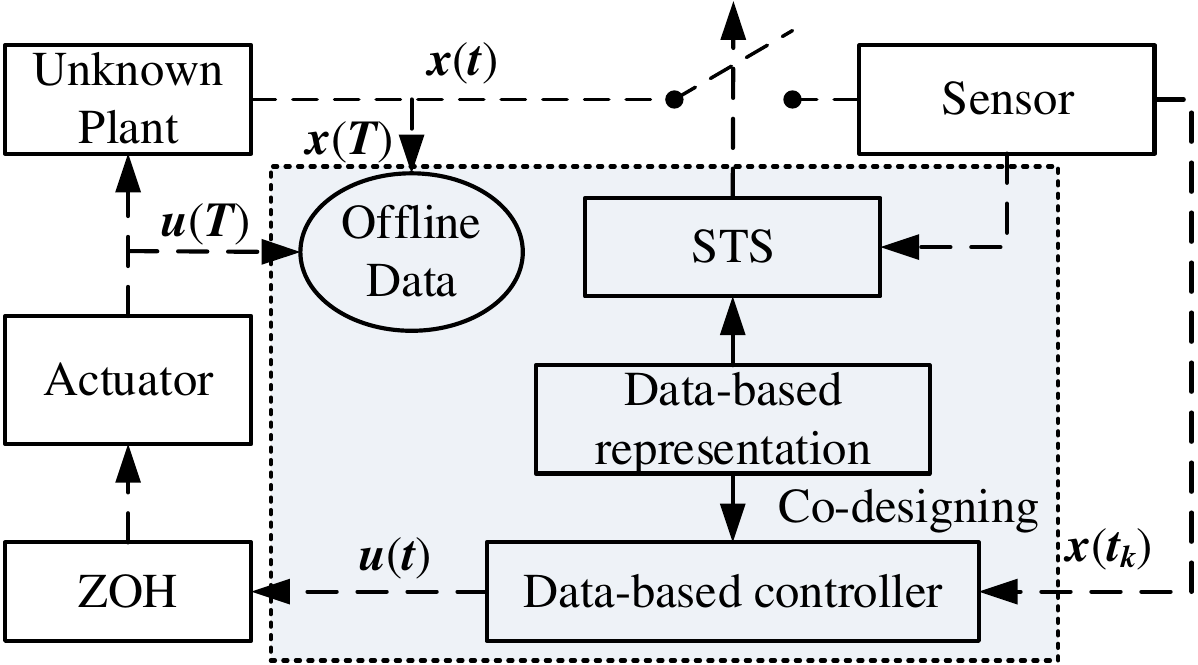}
	\caption{Structure of data-driven discrete-time systems under STS.}
	\label{FIG:structure:self}
\end{figure}

\begin{figure}
	\centering
\includegraphics[scale=0.4]{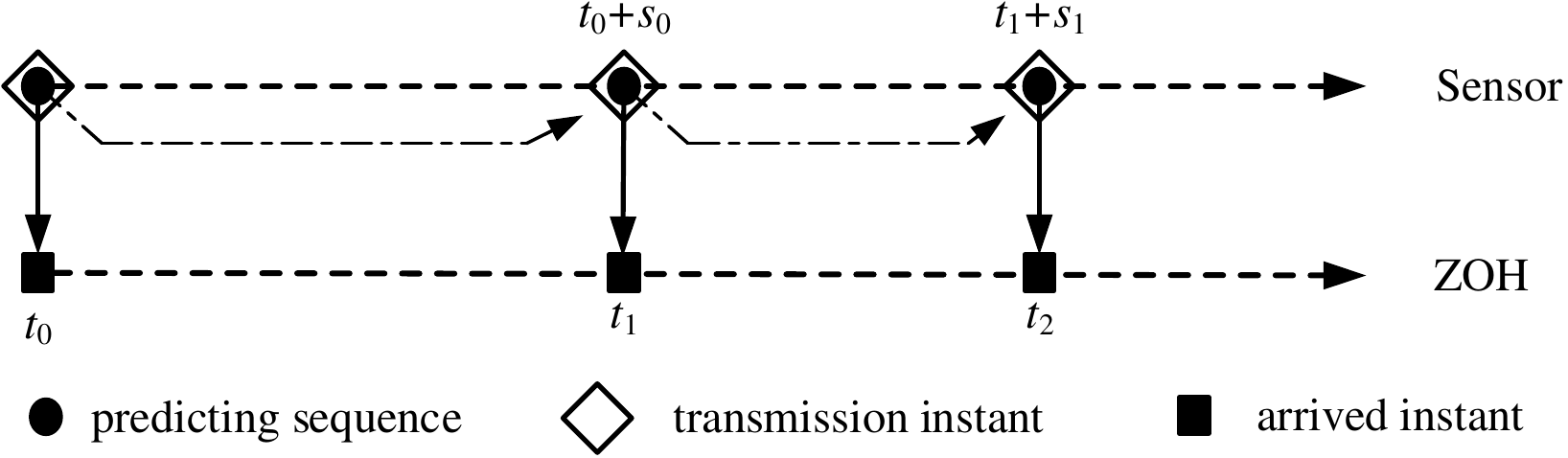}
	\caption{Evolution of transmission series.}
	\label{FIG:evolution:STC}
\end{figure}

The idea of the proposed STS is to build a function $\Gamma(x(t_k))$ for computing the next transmission instant $t_{k+1}$ based on the current state $x(t_k)$ of system \eqref{sys:sampling}
\begin{equation}\label{self:function}
t_{k+1}=t_k+\Gamma(x(t_k),s_k).
\end{equation}
We employ the following condition to find  $\Gamma(x(t_k),s_k)$
\begin{align}
&\sigma_1 x^{\top}(t_k\!+\!s_k)\Omega x(t_k\!+\!s_k)\!+\!\sigma_2 x^{\top}(t_k)\Omega x(t_k) \!-\!e^{\top}(s_k)\Omega e(s_k)\notag\\
\label{sys:self:condition}
&\geq 0
\end{align}
where
$\Omega\succ0$ is some weight matrix; $\sigma_1$ and $\sigma_2$ are parameters to be designed; $e(s_k):=x(t_k+s_k)-x(t_k)$ denotes the error between the sampled signals $x(t_k)$ at the latest transmission
instant and $x(t_k+s_k)$ at time $t_k+s_k$.
According to \eqref{sys:switch}, the condition \eqref{sys:self:condition}
can be reformulated in the form of a QMI
\begin{align}\label{sys:QMI}
\mathcal{Q}(x(t_k),s_k)= &
\left[\!\!\begin{array}{cc}(A^{s_k}+\underline{B}^{s_k}\underline{K}^{s_k}) x(t_k)\\x(t_k) \\\end{array}\!\!\right]^{\top}  \notag \\
&\times \left[\!\!\begin{array}{cc}(\sigma_1-1)\Omega & \!\!\Omega\\\ast & \!\!(\sigma_2-1)\Omega \\\end{array}\!\!\right][\cdot]\geq 0.
\end{align}
If \eqref{sys:QMI} is satisfied, the time $t_k+s_k$ is declared to be the next transmission instant, i.e., $t_{k+1}=t_k+s_k$. When the sampled state $x(t_k+s_k)$ is transmitted to the controller, a ZOH is used to maintain it within the
interval $[t_{k+1}, t_{k+2}-1]$. Simultaneously,
the self-triggering module is updated and employed to predict the next transmission instant $t_{k+2}$.
Consequently, the function $\Gamma(x(t_k))$ is designed to be
\begin{equation}\label{sys:self:function}
\Gamma(x(t_k))=\max_{s_k\in \mathbb{N}}\Big\{s_k>0\Big|\mathcal{Q}(x(t_k),s_k)\geq 0\Big\}.
\end{equation}
In Fig. \ref{FIG:evolution:STC}, an example illustrating the STS is given.
The next transmission instant is predicted using the current transmitted measurement; that is, $t_{k+1}$ is determined by the system state at time $t_k$.

\begin{Remark}[\emph{Relationship between ETS and STS}]\label{ETStoSTS}
\emph{The STS condition in \eqref{sys:self:condition} is a static
triggering scheme that is similar to \cite{Liu2015,Tripathy2017Discrete}, discussed in
Remark \ref{general}.
When the parameter $\theta$ approaches infinity and $jh=s_k$, the ETS condition in \eqref{sys:trigger} boils down to
\eqref{sys:self:condition}.
Both STS and ETS are based on the current transmitted signal $x(t_k)$. However, the next transmission instant in STS is
determined by predicted states, while the ETS uses currently sampled data. 
Extending the above STS by considering a dynamic triggering scheme is an interesting issue for future research.}
\end{Remark}

\subsection{Data-driven STS}
Although the system matrices $A$ and $B$ are \emph{unknown},
the measurements $\{x(T)\}^{\rho+s}_{T=0}$, $\{u(T)\}^{\rho+s-1}_{T=0}$ $(T,\rho,s\in\mathbb{N},\rho>0,s>0)$ of the perturbed system \eqref{sys:data:perturbed} are available. Based on the model-based function in \eqref{sys:self:function},
a data-based discrete-time STS is proposed. 
Our idea is to rebuild a self-triggering function using the data $\{x(T)\}^{\rho+s}_{T=0}$, $\{u(T)\}^{\rho+s-1}_{T=0}$ to replace the $[A^{s}~\underline{B}^{s}]$-based representation \eqref{sys:self:function}.
To that end, we recall the data-driven parametrization of the lifted matrix $[A^s~\underline{B}^{s}]$ in \cite{wildhagen2021datadriven}. Similar to the system expression in \eqref{sys:switch}, we firstly re-write the perturbed system \eqref{sys:data:perturbed} as follows 
\begin{align}\label{sys:data:perturbed:self}
&x(T+s)=A^s x(T)+\underline{B}^s
\left[\begin{array}{cc}
    u(T) \\
    \vdots \\
    u(T+s-1)\\
  \end{array}\right]\notag \\
  &~~~~~~~~~+
  \left[\!\begin{array}{ccc}
    A^{s-1}B_w & \cdots & B_w
  \end{array}\!\right]
  \left[\!\begin{array}{cc}
    w(T) \\
    \vdots \\
    w(T+s-1)\\
  \end{array}\!\right].
\end{align}
Recall that the measured data $\{x(T)\}^{\rho+s-1}_{T=0}$, $\{u(T)\}^{\rho+s-2}_{T=0}$ are corrupted by the {\it unknown} noise  $\{w(T)\}^{\rho+s-2}_{T=0}$.
Let us define the following matrices containing the measurements
\begin{align*}
X_+^s &:=\left[\begin{array}{cccc}x(s)&x(1+s)&\cdots &x(\rho+s-1) \\\end{array}\right]\\
U^s&:=\left[\begin{array}{cccc}
    u(0)&u(1)&\cdots& u(\rho-1) \\
    \vdots& \vdots& & \vdots\\
    u(s-1)&u(s)&\cdots& u(\rho+s-2)\\
  \end{array}\right].
  \end{align*}
We further define the following lifted disturbance
  \begin{align*}
W^1&:=\left[\begin{array}{cccc}
    w(0)&w(1)&\cdots& w(\rho-1)
     \end{array}\right]\\
\underline{W}^s&:=
\left[\begin{array}{cccc}
    w(0)&\cdots& w(\rho-1) \\
    \vdots& & \vdots\\
    w(s-1)&\cdots& w(\rho+s-2)\\
     \end{array}\right]\\
    W^s&:=\left[\begin{array}{ccc}
    A^{s-1}B_w & \cdots &B_w
  \end{array}\right] \underline{W}^s,~~\text {\it{for}}~ s> 1.
\end{align*}
Then, it is clear that
\begin{align}\label{formu:data:self}
X_+^s = A^s X +\underline{B}^s U^s +B^s_w W^s
\end{align}
where $B^1_w:=B_w$, $B^s_w:=I$ for $s>1$.
Similar to Assumption \ref{Ass:disturbance}, we make the following assumption on the noise. 
\begin{Assumption}[\emph {Lifted noise bound}]\label{Ass:disturbance:self}
\emph{The noise sequence $\{w(t)\}^{\rho+s-2}_{t=0}$ collected in the matrix $W^s$ satisfies $W^s\in\mathcal{W}^s$ with}
\begin{equation*}
\begin{aligned}
\mathcal{W}^s=\bigg\{W^s\in\mathbb{R}^{n_w^s\times\rho} \Big |
\left[\!\begin{array}{cc}{W^s}^{\top}\\I \\\end{array}\!\right]^{\top}\!\!
  \left[\!\begin{array}{cc}Q_d^s\! & \!S_d^s\\\ast\! & \!R_d^s \\\end{array}\!\right]\!\!
  \left[\!\begin{array}{cc}{W^s}^{\top}\\I \\\end{array}\!\right]\succeq0 \bigg\}
\end{aligned}
\end{equation*}
\emph{for some known matrices $Q_d^s \prec 0 \in \mathbb{R}^{\rho\times \rho}$, $S_d^s \in \mathbb{R}^{\rho\times n_w^s}$, and $R_d^s={R_d^{s}}^{\top} \in \mathbb{R}^{n_w^s\times n_w^s}$, where
$n_w^1:=n_w$, $n_w^s:=n$ for $s>1$.}
\end{Assumption}

Define the set of all pairs $[A^s~\underline{B}^s]$ consistent with the model \eqref{formu:data:self} and Assumption \ref{Ass:disturbance:self} as the same as \cite{wildhagen2021datadriven}
\begin{align}\label{sys:data:self}
\Sigma_{AB}^s&:=\{[A^s~ \underline{B}^s]\in \mathbb{R}^{n \times (n+sm)} \mid \notag \\
 &~~~~~ X_+^s=A^sX+\underline{B}^sU^s+B_w^s W^s,~ W^s\in \mathcal{W}^s\}.
\end{align}
Analogously to Lemma \ref{Lemma:system:data}, we obtain the following equivalent expression of $\Sigma_{AB}^s$ in the form of a QMI
\begin{equation}\label{data:represent:self}
\begin{aligned}
\Sigma_{AB}^s=~&\bigg\{[A^s~ \underline{B}^s]\in \mathbb{R}^{n \times (n+sm)} \Big | \\
&~\left[\begin{array}{cc}[A^s~ \underline{B}^s]^{\top}\\I \\\end{array}\right]^{\top}
  \Theta_{AB}^s
  \left[\begin{array}{cc}[A^s~ \underline{B}^s]^{\top} \\
   I\\\end{array}\!\right]\succeq0
\bigg\}
\end{aligned}
\end{equation}
where
\begin{align*}
\Theta_{AB}^s=\left[\begin{array}{cc}Q_c^s & S_c^s\\\ast & R_c^s \\\end{array}\right]:=
\left[\begin{array}{cc}-X & 0 \\ -U^s & 0 \\ \hline X_+^s & B_w^s \\\end{array}\right]
  \left[\begin{array}{cc}Q_d^s & S_d^s\\\ast & R_d^s \\\end{array}\right]
  [\cdot]^{\top}.
  \end{align*}
Having obtained a data-based representation of system \eqref{sys:switch}, we can now translate the
model-based self-triggering function  \eqref{sys:self:function} that depends on $[A^{s}~\underline{B}^{s}]$  to a data-based one.
The following technical assumption on the matrix $\Theta_{AB}^s$ is required for the subsequent derivation.

\begin{Assumption}\label{Ass:datamatrix}
\emph{The matrix $\Theta_{AB}^s$ is invertible and has $n_w$ positive eigenvalues.}
\end{Assumption}
In practice, Assumptions \ref{Ass:datamatrix} is satisfied when the data are sufficiently rich and $B_w$ is invertible \cite{Berberich2020}.
Based on Assumption \ref{Ass:datamatrix}, we have the following theorem.


\begin{Theorem}[\emph{Data-driven self-triggering condition}]\label{data:self:scheme}
For given scalars $\sigma_{1}\geq0$, $\sigma_{2}\geq0$, matrix $\Omega\succ0$, controller gain $K$, and $x(t_k)$ from  system \eqref{sys:switch}, $\mathcal{Q}(x(t_k),s)$ in \eqref{sys:QMI} satisfies
\begin{equation}\label{th3:Q}
\mathcal{Q}(x(t_k),s)\geq0
\end{equation}
for any $[A^{s}~\underline{B}^{s}]\in \Sigma_{AB}^{s}$, if there exists a scalar $\gamma>0$, such that the following LMI holds for some  $s\in\mathbb{N}, s\geq1$
\begin{equation}\label{Th3:LMI}
\tilde{\mathcal{Q}}(x(t_k))-\gamma \tilde{\mathcal{G}}^{s}(x(t_k)) \succeq 0
\end{equation}
where
\begin{align}
\tilde{\mathcal{Q}}(x(t_k))&:=\left[\begin{array}{cc} I & 0\\0 &x^{\top}(t_k) \\\end{array}\right]
\!\left[\!\!\begin{array}{cc}(\sigma_1-1)\Omega & \!\!\Omega\\\ast & \!\!(\sigma_2-1)\Omega \\\end{array}\!\!\right]\!\![\cdot]^{\top}\notag\\
\tilde{\mathcal{G}}^{s}(x(t_k))&:= \left[\begin{array}{ccc} I & 0 & 0\\0 &x^{\top}(t_k)& x^{\top}(t_k)\underline{K}^{s\top}\\\end{array}\right] \tilde{\Theta}_{AB}^{s} [\cdot]^{\top}\notag\\
\tilde{\Theta}_{AB}^{s}&:=\left[\!\!\begin{array}{cc}-\tilde{R}_c^{s} & {\tilde{S}_c^{s\top}}\\\ast &  -\tilde{Q}_c^{s} \\\end{array}\!\!\right],\left[\!\begin{array}{cc}\tilde{Q}_c^s & \tilde{S}_c^{s}\\\ast & \tilde{R}_c^{s} \\\end{array}\!\right]:=
\left[\!\begin{array}{cc}\!Q_c^{s} & S_c^{s}\!\\ \!\ast & R_c^{s}\! \\\end{array}\!\right]^{-1}\notag
\end{align}
\end{Theorem}
\begin{proof}
The matrix $\mathcal{Q}(x(t_k),s)$ in \eqref{sys:QMI} is rewritten as
\begin{align}\label{sys:QMI:proof}
\mathcal{Q}(x(t_k),s)=
\left[\!\!\begin{array}{cc}(A^{s}+\underline{B}^{s} \underline{K}^{s}) x(t_k)\\I \\\end{array}\!\!\right]^{\top} \tilde{\mathcal{Q}}(x(t_k)) [\cdot].
\end{align}
Applying the dualization lemma \cite[Lemma 4.9]{Scherer2000} to the system representation in \eqref{data:represent:self} under Assumption \ref{Ass:datamatrix}, it can be proven that $[A^{s}~\underline{B}^{s}]\in \Sigma_{AB}^{s}$ if and only if
\begin{align}
\left[\!\!\begin{array}{cc}[A^{s}~\underline{B}^{s}]\\I \\\end{array}\!\!\right]^{\top}
\tilde{\Theta}_{AB}^s
\left[\!\!\begin{array}{cc}[A^{s}~\underline{B}^{s}]\\I \\\end{array}\!\!\right]
\succeq0.
\end{align}
Immediately, through the full-block S-procedure \cite{Sche2001}, we have $\tilde{Q}(x(t_k)) \succeq 0$ for any $[A^{s}~\underline{B}^{s}]\in \Sigma_{AB}^{s}$ if there exists a scalar $\gamma>0$ such that the LMI \eqref{Th3:LMI} holds.
End the proof.
\end{proof}

It is possible to derive LMI-based conditions for guaranteeing stability of both the model-based and the data-driven STS approach, and to co-design the controller gain and the triggering matrix. These results are omitted for space reasons.

\begin{Remark}[\emph{Explanation of the data-driven STS condition}]
\emph{Theorem \ref{data:self:scheme} offers a data-driven triggering condition based on the model-based one in
\eqref{sys:self:function}.
The key idea is that we
leverage the data-based representation in \eqref{data:represent:self} to robustly verify the STS condition for all $[A^s~ \underline{B}^s]$ consistent with the data.
As a result, the model-based triggering STS function in \eqref{sys:self:function} is translated into a
data-driven one as follows
\begin{equation}\label{sys:self:function:data}
\bar{\Gamma}(x(t_k),s_k)=\max_{s_k \in \mathbb{N}}\Big\{s_k\geq1\Big|\tilde{\mathcal{Q}}(x(t_k))-\gamma \tilde{\mathcal{G}}^{s_k}(x(t_k))\succeq 0\Big\}.
\end{equation}
Overall, our data-driven STS is given by
\begin{equation}\label{self:function:data}
t_{k+1}=t_k+\bar{\Gamma}(x(t_k),s_k)
\end{equation}
under Assumptions \ref{Ass:disturbance:self} and \ref{Ass:datamatrix} for system \eqref{sys:switch}.
Note that, in Theorem \ref{data:self:scheme}, the data-driven condition $\tilde{\mathcal{Q}}(x(t_k))-\gamma \tilde{\mathcal{G}}^{s}(x(t_k)) \succeq 0$  sufficiently guarantees the model-based one $\mathcal{Q}(x(t_k),s)\geq0$ in \eqref{sys:self:function} that is consistent with system stability characteristics.
A smaller triggering interval may be produced by \eqref{sys:self:function:data}, since
\eqref{sys:self:function:data} only provides a sufficient 
condition for \eqref{sys:self:function}.}
\end{Remark}
%

\begin{Remark}[\emph{Summary of data-driven STS algorithm}]
\emph{According to \eqref{self:function:data}, the next transmission instant $t_{k+1}$ of system \eqref{sys:sampling} can be computed using only collected data $\{x(T)\}^{\rho+s-1}_{T=0}$ and $\{u(T)\}^{\rho+s-2}_{T=0}$.
Note that the matrix $\tilde{\Theta}_{AB}^{s}$ in \eqref{self:function:data} needs to be determined in advance from the given noise bound.
To that end, we recall \cite[Algorithm 1]{wildhagen2021datadriven}, which can be used to construct a lifted noise bound as in Assumption \ref{Ass:disturbance:self} based on pointwise bound $||w(T)||_2\leq\bar{w}$ for all $T=0,\,\dots,\,\rho+s-2$ with some $\bar{w}>0$. 
This leads to a lifted system parametrization as in \eqref{data:represent:self}.
Then, we continuously check the data-driven self-triggering condition using the matrices $\tilde{\Theta}_{AB}^{s}$ from \cite[Algorithm 1]{wildhagen2021datadriven}.
The next triggering instant can be determined by checking the LMI \eqref{Th3:LMI} as soon as the current transmission instant and state become available.}
\end{Remark}

\begin{Remark}[\emph {Motivation}]\label{remark:motivation}
\emph{This paper combines the data-driven control and ETS/STS for discrete-time sampled-data systems, where the co-design problem of the controller and the triggering matrix without any knowledge of the system model is solved. 
To the author's knowledge, the only alternative to this approach in the current literature would be system identification, e.g., least-squares estimation of the system matrices \cite{Oymak2019identification}, followed by discrete-time ETS/STS in references \cite{Ding2020,Hu2016}.
However, while the proposed approach guarantees that the closed-loop system under the designed ETS/STS is stable, such an identification-based ETS/STS does in general not provide such guarantees, in particular when the data are affected by noise.
This is due to the fact that 1) providing tight estimation bounds in system identification is in challenging, compare, e.g., \cite{Matni2019}; and 2) \cite{Ding2020,Hu2016} only provide nominal results, i.e., error bounds arising from identification based on noisy data are not handled systematically.}
\end{Remark}

\section{Example and Simulation}\label{sec:example}
In this section, two numerical examples from \cite{Zhang2001,Hu2016} are employed to certificate the effectiveness and merits of our proposed methods.  All numerical computations were performed using Matlab, together with the SeDuMi toolbox  \cite{SeDuMi}.

\begin{Example}\label{ex:1}
\emph{Consider the linear system used in \cite{Zhang2001}
\begin{equation*}
\dot{x}(\nu)=\left[
  \begin{array}{cc}
    0& 1 \\
    0 & -0.1 \\
  \end{array}
\right] x(\nu)+\left[
  \begin{array}{cc}
    0 \\
    0.1\\
  \end{array}
\right] u(\nu),~\nu\geq0.
\end{equation*}
Discretizing the system leads to
\begin{equation}\label{example:discrete:system}
x(t+1)=A(T_k) x(t)+B(T_k) u(t), ~t\in \mathbb{N}
\end{equation}
where $T_k>0$ is a discretization interval, and
\begin{equation*}
A(T_k):=e^{\left[
  \begin{array}{cc}
    0& 1 \\
    0 & -0.1 \\
  \end{array}
\right]T_k},
B(T_k):=\int_0^{T_k}e^{A(s)}{\left[
  \begin{array}{cc}
    0 \\
    0.1\\
  \end{array}
\right]}ds.
\end{equation*}
We assume that the linear sampled-data state-feedback controller $u(t)=Kx(t_k)$ is used to control the system. For $t\in \mathbb{N}_{[t_k, t_{k+1}-1]}$, the system \eqref{example:discrete:system} can be written as
\begin{equation}\label{example:sample:system}
x(t+1)=A(T_k) x(t)+B(T_k) Kx(t_k).
\end{equation}}
\end{Example}
\begin{table}[t]
\caption{Maximum allowable $\bar{h}$ with $\underline{h}=1$ for different noise bounds $\bar{w}$ and given $K=-[3.75~ 11.5]$.}
\begin{center}      
\setlength{\tabcolsep}{4pt}
\renewcommand\arraystretch{1.}
\begin{tabular}{lcccccccccc}
\hline\noalign{\smallskip}
$\bar{w}$ & 0.001 & 0.002 & 0.005 & 0.01&0.02& 0.05 \\
\noalign{\smallskip}\hline\noalign{\smallskip}
\cite[Theorem 11]{wildhagen2021datadriven}& 12& 12 & 11 &11&-&-\\
\cite[Theorem 20]{wildhagen2021datadriven}& 16& 15 & 1 &-&-&-\\
\cite[Theorem 3]{wildhagen2021improved}&13 & 13& 13& 12& -&-\\
Theorem \ref{Th2}&17 &16& 15&13&12&8\\
\hline\noalign{\smallskip}
\end{tabular}
\end{center}
\label{Tab:noise:givenK}
\end{table}
\begin{table}[t]
\caption{Maximum allowable $\bar{h}$ with $\underline{h}=1$ for different noise bounds $\bar{w}$ by optimizing controller gain $K$.}
\begin{center}      
\setlength{\tabcolsep}{4pt}
\renewcommand\arraystretch{1.}
\begin{tabular}{lccccccccc}
\hline\noalign{\smallskip}
$\bar{w}$ & 0.001 & 0.002 & 0.005 & 0.01&0.02& 0.05\\
\noalign{\smallskip}\hline\noalign{\smallskip}
\cite[Corollary 13]{wildhagen2021datadriven}& 15& 11 & 7 &5&3&-\\
\cite[Corollary 23]{wildhagen2021datadriven}& 19& 17 &12 &8&5&1\\
Theorem \ref{Th2} &54 &45 &34 &27 &16 &9\\
\noalign{\smallskip}\hline
\end{tabular}
\end{center}
\label{Tab:noise:unK}
\end{table}
(\emph{Upper bounds of $h$ for known $A$ and $B$}) Similar to \cite{wildhagen2021datadriven,wildhagen2021improved}, the controller gain $K=-[3.75 ~ 11.5]$ is employed.
By setting $\sigma_1=0$, $\sigma_2=0$, with $\theta\to \infty$ (that is, $\eta({\tau_j^k})=0$), the proposed triggering scheme \eqref{sys:trigger} reduces to a periodic transmission scheme. We first consider the
model-based stability analysis. Leveraging Theorem
\ref{Th1} with discretization interval $T_k=0.01$, the maximum sampling interval $\bar{h}$ leading to closed-loop stability is $\bar{h}=173$ with $\underline{h}=1$. Compared with the model-based results of $h=111$ obtained by \cite{Nag2008}, $\bar{h}=122$ by \cite{wildhagen2021datadriven}, $\bar{h}=133$ by \cite{Seuret2018discrete}, $\bar{h}=136$ by \cite{wildhagen2021improved}, and $\bar{h}=169$ by \cite{Fridman2010}, Theorem \ref{Th1} in our paper provides improvements of $55.8\%$, $41.8\%$, $30.0\%$, $27.2\%$, and $2.3\%$, respectively. This shows the merits of the proposed DLF approach in reducing the conservatism of stability conditions.

(\emph{Upper bounds of $h$ for unknown $A$ and $B$}) Next, assume that the matrices $A$ and $B$ are \emph{unknown}. We set the discretization interval as $T_k=0.1$
and generated $\rho=800$
measurements $\{x(T)\}^{\rho}_{T=0}$, $\{u(T)\}^{\rho-1}_{T=0}$, where the data-generating input was sampled uniformly from $u(T)\in [-1,~1]$.
The measured data were perturbed by a disturbance distributed uniformly over $w(T)\in [-\bar{w},\,\bar{w}]^2$ for $\bar{w}>0$.
 Such disturbance $w(T)$ fulfills Assumption \ref{Ass:disturbance} with $Q_d=-I$, $S_d=0$, and $R_d=\bar{w}^2\rho I$ $(\rho=800)$.
The matrix $B_w$ was taken as $B_w=0.01I$, which has full column rank.
Using Theorem \ref{Th2} and $K=-[3.75~ 11.5]$ and setting parameters $\sigma_1=0$, $\sigma_2=0$, and $\epsilon=2$, values of $\bar{h}$ with $\underline{h}=1$  for different realizations of $\bar{w}$ were computed and presented in Table \ref{Tab:noise:givenK}.
The results come from different approaches in \cite{wildhagen2021datadriven,wildhagen2021improved} under the same levels of disturbance are also given. According to the comparison in Table \ref{Tab:noise:givenK},
Theorem \ref{Th2} provides larger values of $\bar{h}$, 
i.e., our method reduces the conservatism if compared to
\cite[Theorem 11]{wildhagen2021datadriven}, \cite[Theorem 20]{wildhagen2021datadriven}, and \cite[Theorem 3]{wildhagen2021improved}.
Furthermore, we design a controller leading to a possibly large sampling bound $\bar{h}$ with $\underline{h}=1$. By virtue of \cite[Corollary 13]{wildhagen2021datadriven}, \cite[Corollary 23]{wildhagen2021datadriven}, and Theorem \ref{Th2}, corresponding values of $\bar{h}$ were computed and listed in Table \ref{Tab:noise:unK}. Again, Theorem \ref{Th2} provides the largest $\bar{h}$, which leads to the same conclusion as by Table \ref{Tab:noise:givenK}. In the following, the proposed data-driven ETS \eqref{sys:trigger} and STS \eqref{self:function:data} are applied for system \eqref{example:sample:system}, respectively.

\subsection{Data-driven ETS and STS control}
(\emph{Data-driven ETS control}) For data-driven control of system \eqref{sys:sampling} with unknown matrices $A$ and $B$ under our transmission scheme \eqref{sys:trigger}, we now employ Theorem \ref{Th2} to co-design the controller gain $K$ and
the triggering matrix $\Omega$ using the same measurements 
as above. We set the sampling interval $h=2$, triggering parameters $\sigma_1=0.5$, $\sigma_2=0.5$, $\theta=2$, $\lambda=0.2$, and discretization period $T_k=0.1$.
Solving the data-based LMIs \eqref{Th2:LMI1} in Theorem \ref{Th2} with $\epsilon=2$,
the controller and triggering matrices are co-designed as follows
\begin{equation*}
\begin{aligned}
K=[
  \begin{array}{cccc}
  -0.2908 &  -4.0340
  \end{array}],~
\Omega&=\left[
  \begin{array}{cccc}
    0.0001 &   0.0007\\
    0.0007  &  0.0104\\
  \end{array}
\right].
\end{aligned}
\end{equation*}
We then simulate system \eqref{sys:sampling} under the triggering scheme in \eqref{sys:trigger} with initial condition $x(0) = [3 ~ -2]^{\top}$, as well as the dynamic variable $\eta({\tau_j^k})$, for $t\in [0,~400]$. Their trajectories are depicted in Fig. \ref{FIG:event}.
Evidently, both the system states and the dynamic variable converge to the origin,
which demonstrates the feasibility of our designed controller gain $K$ and the triggering matrix $\Omega$. It is also worth pointing out that \emph{only} $10$ measurements were transmitted to the controller under our proposed triggering scheme in \eqref{sys:trigger}, while $200$ measurements were sampled.
This validates the effectiveness of the proposed scheme in saving communication resources, while maintaining stability.

\begin{figure}[t]
\centering
\includegraphics[height=5.1cm,width=6.8cm]{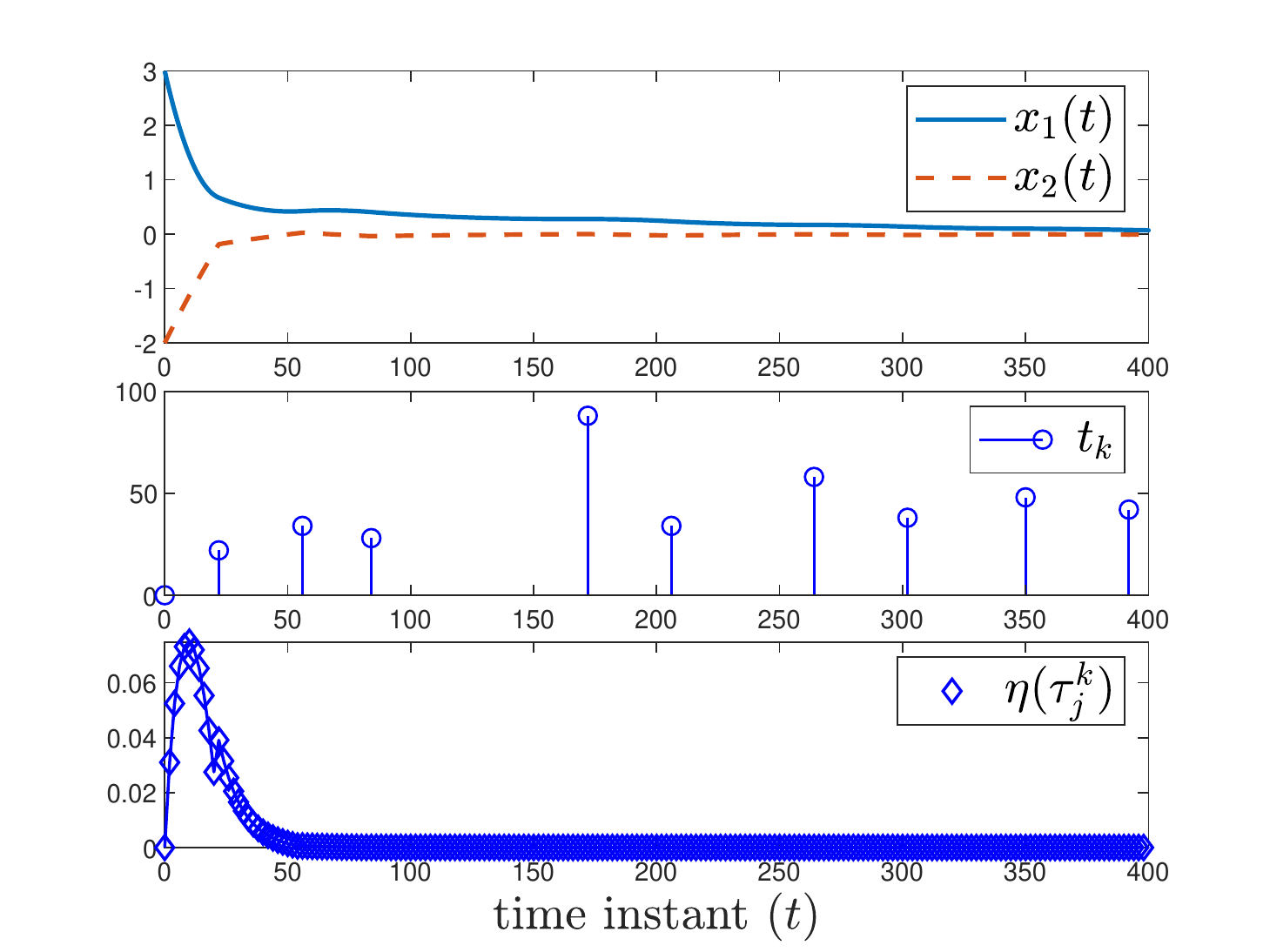}
\caption{Trajectories of system \eqref{sys:sampling} and dynamic variable $\eta({\tau_j^k})$ under data-driven ETS  \eqref{sys:trigger} with the initial condition $x(0) = [3~ -2]^{\top}$.} 
\label{FIG:event}
\end{figure}

\begin{figure}[t]
	\centering
		 \includegraphics[height=5.1cm,width=6.8cm]{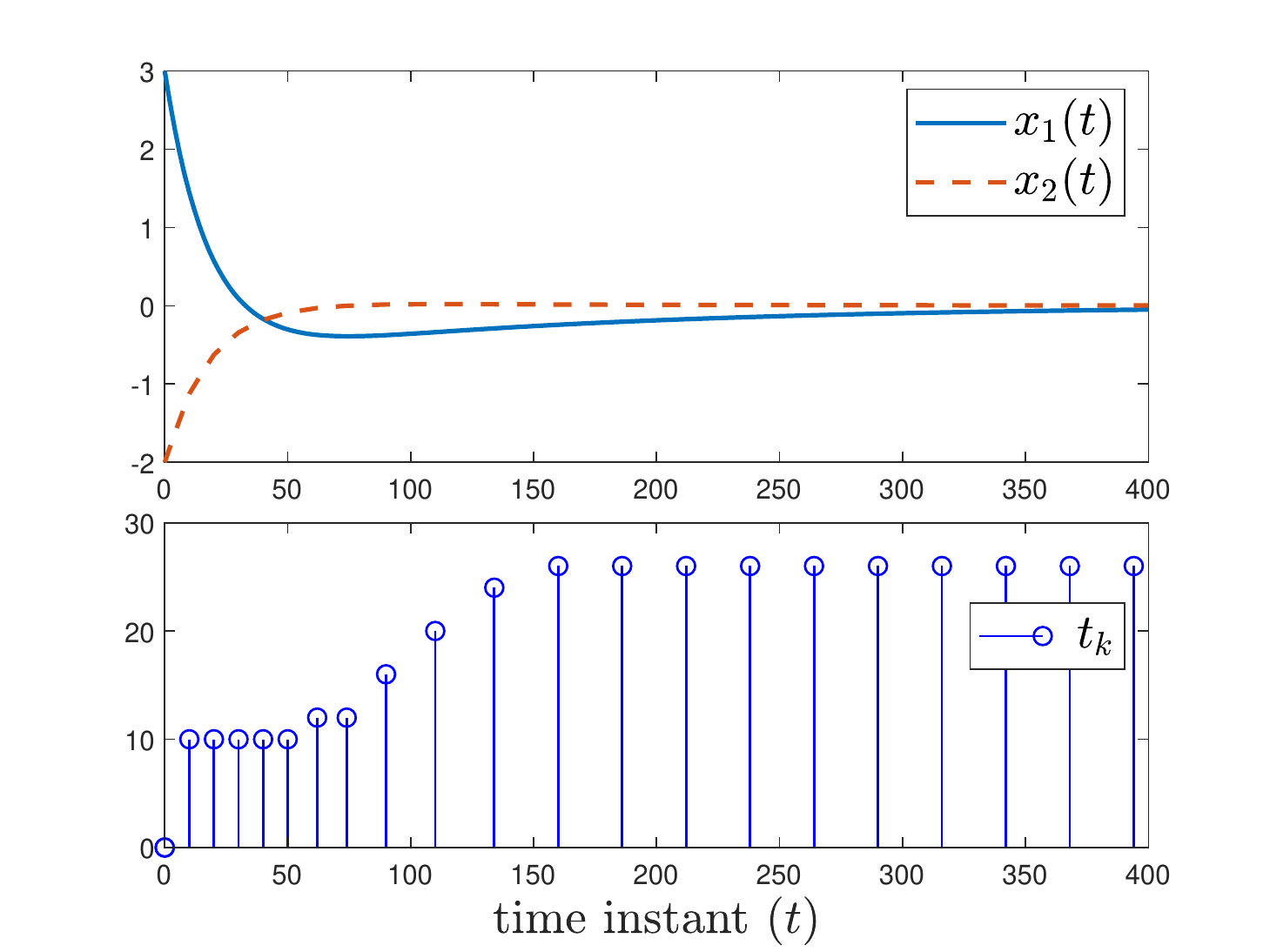}
	\caption{Trajectories of system \eqref{sys:sampling} under data-driven STS \eqref{self:function:data} with 
the initial condition $x(0) = [3~ -2]^{\top}$.}
	\label{FIG:self}
\end{figure}

(\emph{Data-driven STS control}) The data-based STS \eqref{self:function:data} for unknown systems is tested in the following.
Firstly, according to \cite[Algorithm 1]{wildhagen2021datadriven},
the bound on the lifted disturbance in Assumption \ref{Ass:disturbance:self} can be computed by using above measurements
with $\rho=750$ and $\bar{s}=50$. Then,
it is straightforward to obtain
the data-based matrices $\tilde{\mathcal{G}}^{s}(x(t_k))$ ($s=jh$) for $s\in \mathbb{N}_{[1, \,\bar{s}]}$ in \eqref{Th3:LMI} based on the matrix $\Theta_{AB}^s$ from the bound in Assumption \ref{Ass:disturbance:self}.

Applying the same controller gain and the triggering matrix as in the ETS control (which also guarantees the stability under the STS when using the same parameters in the ETS scheme, since as discussed in Remark \ref{ETStoSTS} the STS is a special case of the ETS), the system state trajectory under the data-based STS \eqref{self:function:data} with $s=jh$ is depicted in Fig. \ref{FIG:self} for $t\in [0,~400]$. All states converge to the origin, thereby validating the practicality of our proposed co-design method and STS. Note that, in Fig. \ref{FIG:self}, \emph{only} $22$ out of $200$ samples were transmitted to the controller. This illustrates the usefulness of the STS in reducing transmissions while ensuring stability.
Moreover, note that in Figs \ref{FIG:event} and \ref{FIG:self}, more data were generated by the STS compared with the ETS. The main reason is that the STS law \eqref{self:function:data}
is more conservative than that of ETS
\eqref{sys:trigger}. 

\begin{figure}[t]
\centering
\includegraphics[height=5.1cm,width=6.8cm]{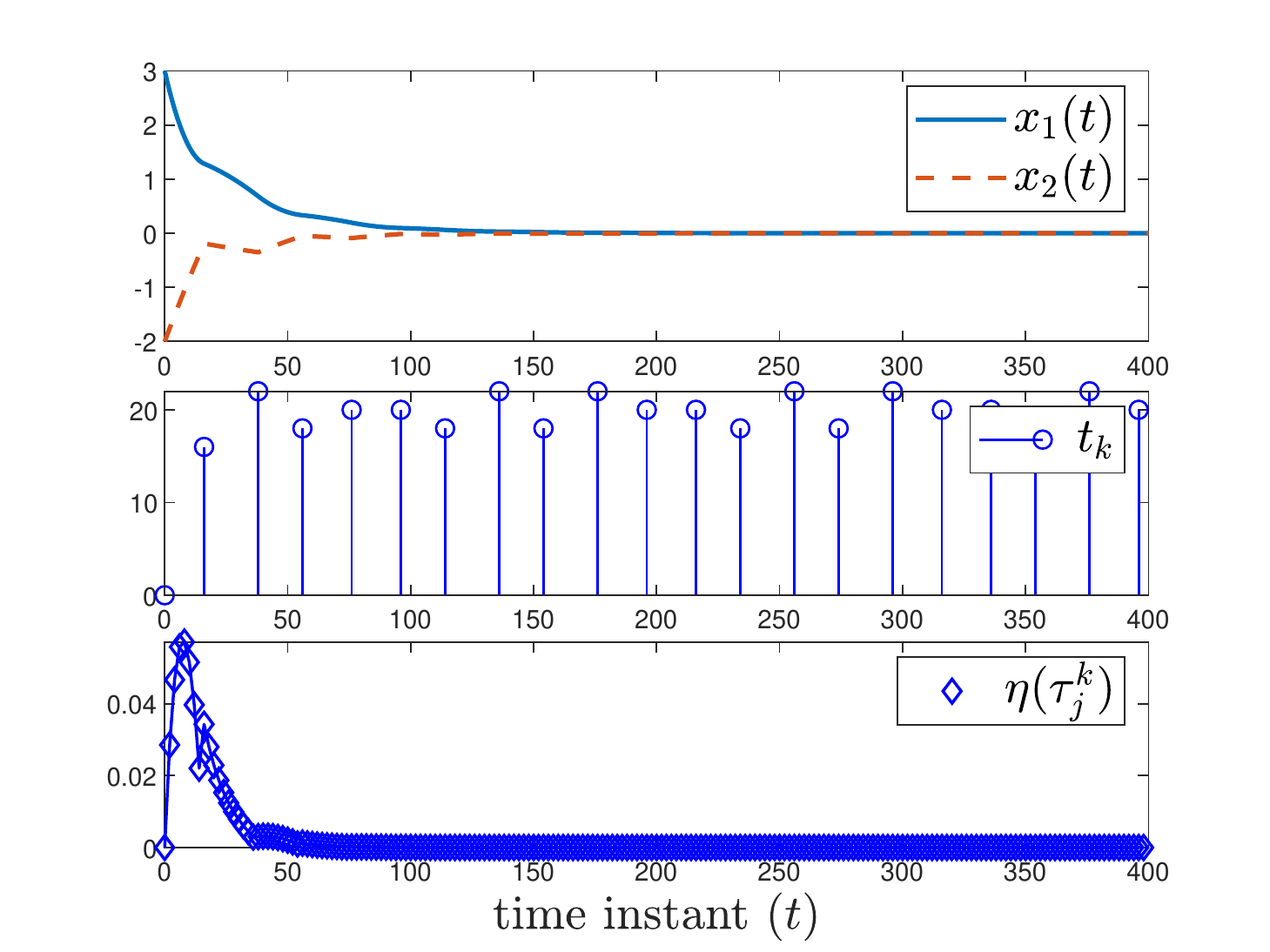}
\caption{Trajectories of system \eqref{sys:sampling} and dynamic variable $\eta({\tau_j^k})$ under identification-based ETS  \eqref{sys:trigger} with 
the initial condition $x(0) = [3~ -2]^{\top}$.}
\label{FIG:event:iden}
\end{figure}

\begin{figure}[t]
	\centering
		 \includegraphics[height=5.1cm,width=6.8cm]{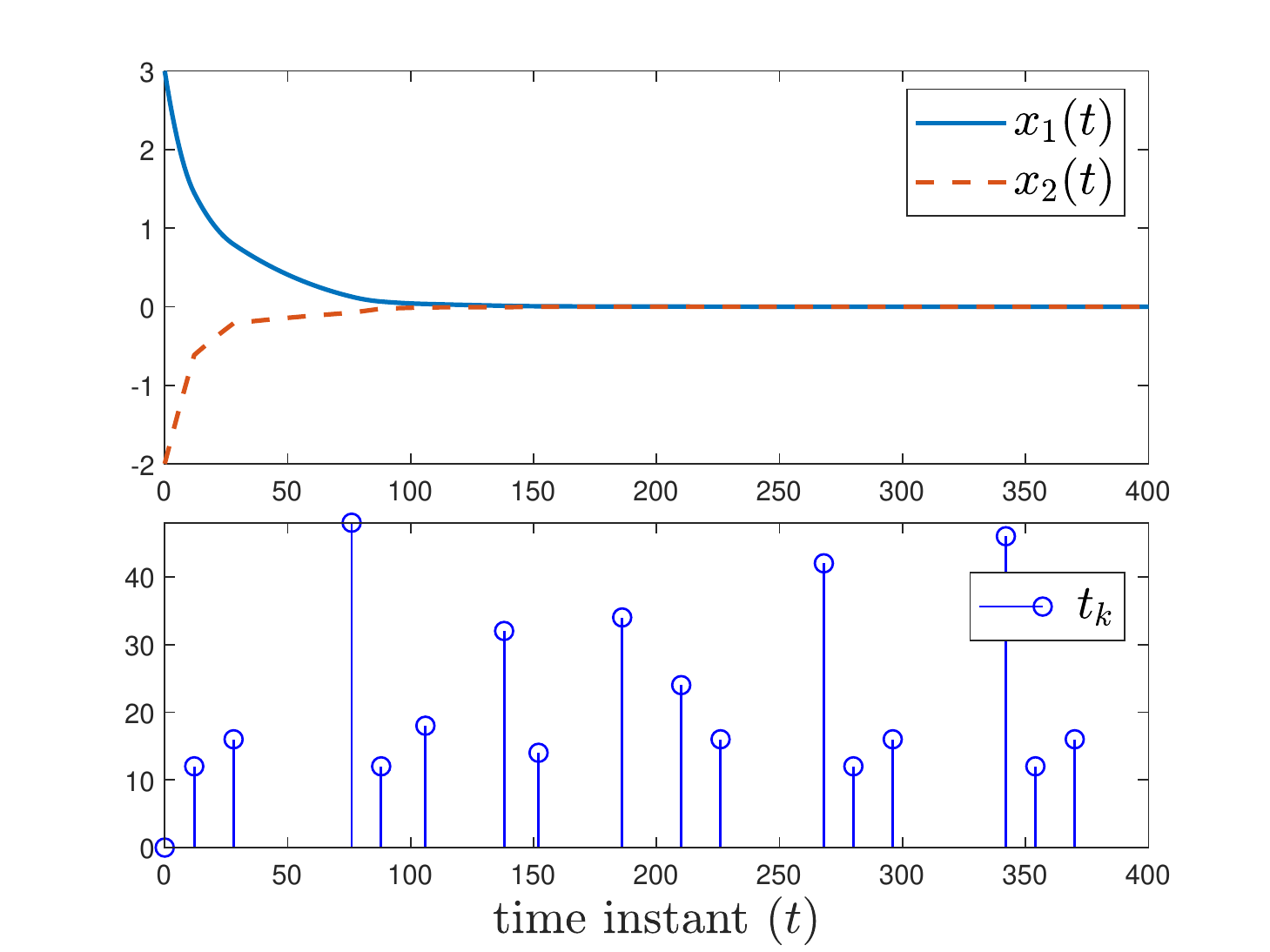}
	\caption{Trajectories of system \eqref{sys:sampling} under identification-based STS \eqref{sys:self:function} with 
the initial condition $x(0) = [3~ -2]^{\top}$.}
	\label{FIG:self:iden}
\end{figure}

\subsection{Comparison with identification-based ETS and STS control}
In this part, we compare the data-driven approaches described above to an alternative approach, consisting of least-squares identification of the system matrices and subsequent model-based ETS and STS as presented in Sections \ref{sec:eventtrigger} and \ref{sec:self}.
In the system identification step,
the following least-squares problem is considered
\begin{equation*}
[\hat{A}~\hat{B}]^{\top}=\arg\min_{(A~B)} \sum_{T=0}^{\rho}\|x(T+1)-Ax(T)-Bu(T)] \|_2^2.
\end{equation*}
The least-squares solution $[\hat{A}~\hat{B}]^{\top}$ is given by
$$[\hat{A}~\hat{B}]^{\top}=([X^{\top} ~U^{\top}]^{\top}[X^{\top} ~U^{\top}])^{-1}[X^{\top} ~U^{\top}]^{\top} X_{+}^{\top}$$
which is the estimation of the real system matrices $A$ and $B$. Next, we will show the
the results of identification-based ETS and STS control.

(\emph{Identification-based ETS and STS control})
Using the least-squares approach and the same $\rho=800$ measurements as in the data-driven control, the system matrices are estimated as
\begin{equation*}
\hat{A}=\left[
  \begin{array}{cc}
     1.0000 &   0.0995\\
    0.0000  &  0.9900
  \end{array}
\right],~
\hat{B}=\left[\begin{array}{cc}
     0.0005\\
    0.0100
  \end{array}
\right].
\end{equation*}
Subsequently, by solving the model-based approach in Remark \ref{remark:model:design} with the same parameters as in data-driven control, the controller and the triggering matrices
were computed as follows
\begin{equation*}
\begin{aligned}
K=[
  \begin{array}{cccc}
  -1.9775 &  -8.0873
  \end{array}],~
\Omega&=\left[
  \begin{array}{cccc}
    0.0011  &  0.0037\\
    0.0037  &  0.0168\\
  \end{array}
\right].
\end{aligned}
\end{equation*}
The trajectories of system \eqref{sys:trigger}
under the identification-based ETS \eqref{sys:trigger} and STS \eqref{sys:self:function} were depicted in Figs. \ref{FIG:event:iden} and \ref{FIG:self:iden}, respectively, over $t\in[0,400]$. The simulation results show that $x(t)$ approaches to zero as $t\to\infty$
under the identification-based ETS or STS with the above designed $K$ and $\Omega$. In light of the comparisons with the simulation results of the data-driven ETS and STS, more data ($21$ out of $200$ samples) were transmitted to the controller under the identification-based ETS as in Fig. \ref{FIG:event:iden} than the ones in Fig. \ref{FIG:event}, and almost the same amount of data ($17$ out of $200$ samples) was generated under the STS in Fig. \ref{FIG:self:iden} compared to the one in Fig. \ref{FIG:self}, while less
settling steps ($t=150$) were required to stabilize the system for the identification-based approaches. However, in contrast to the direct data-driven design, the
identification-based ETS and  STS approaches do not provide stability guarantees.

\section{Concluding Remarks}\label{sec:conclude}
In this paper, we proposed data-based event- and self-triggering transmission schemes for discrete-time systems leveraging a novel looped-functional approach.
We also developed methods for co-designing the controller gain and the triggering matrix for the discrete-time ETS and STS systems.
Finally, a numerical example was presented to corroborate the role of our triggering schemes in saving communication resources, as well as the merits and effectiveness of our co-designing methods.

	\bibliographystyle{IEEEtran}
	\bibliography{cas-refs}

\end{document}